\documentclass[journal]{IEEEtran}
\usepackage[english]{babel}
\usepackage{amsmath, amssymb, amsfonts, amsthm, bm}
\usepackage{mathrsfs}
\usepackage{enumerate}
\usepackage{float}
\usepackage{color,graphicx,epstopdf}
\graphicspath{{./eps}{./png}{./figures}}
\usepackage{dsfont}
\usepackage{hyperref}
\usepackage{cite} 
\usepackage[linesnumbered,commentsnumbered,ruled]{algorithm2e}
\usepackage{pstricks,pst-grad}
\usepackage[font=small]{caption}
\usepackage{subcaption}
\usepackage{tikz}

\newcommand{\E}{\mathsf{E}}

\newcommand*\mcupinn[2]{\vcenter{\hbox{$\mathsurround=0pt
  \ifx\displaystyle#1\textstyle\else#1\fi\bigcup$}}}

\newcommand\rv[1]{{\color{black}#1}}

\newtheorem{theorem}{Theorem}

\newtheorem{remark}{Remark}
\newtheorem{proposition}{Proposition}
\newtheorem{assumption}{Assumption}
\newtheorem{property}{Property}

\usepackage{balance}

\usepackage{enumitem}
\setlist[itemize]{align=parleft,left=0pt..1em}

\begin{document}

\title{\bf Average Communication Rate for Networked  Event-Triggered Stochastic Control Systems}
\date{}
\author{Zengjie Zhang, 
Qingchen Liu$\textsuperscript{*}$,
Mohammad H. Mamduhi,
and Sandra Hirche

\thanks{\textsuperscript{*}Corresponding author.}

\thanks{Z. Zhang is with the Department of Electrical Engineering, Eindhoven University of Technology, Netherlands (z.zhang3@tue.nl).}
\thanks{Q. Liu is with the Department of Automation, University of Science and Technology of China, Hefei, China (qingchen\_liu@ustc.edu.cn).}
\thanks{M. H. Mamduhi is with the Automatic Control Laboratory, ETH Zürich, Switzerland (mmamduhi@ethz.ch).}
\thanks{S. Hirche is with the Chair of Information-Oriented Control, Technical University of Munich, Munich, Germany (hirche@tum.de).}
}

\maketitle

\begin{abstract}
Quantifying the average communication rate (ACR) of a networked event-triggered stochastic control system (NET-SCS) with deterministic thresholds is challenging due to the non-stationary nature of the system's stochastic processes.
For \rv{a NET-SCS, the nonlinear statistics propagation of the network communication status brought up by deterministic thresholds makes the precise computation of ACR difficult.} 
Previous work used to over-simplify the computation \rv{using a Gaussian distribution without incorporating this nonlinearity, leading to sacrificed precision}. This paper proposes both analytical and numerical approaches to predict the exact ACR for a NET-SCS using a recursive model. 
We use theoretical analysis and a numerical study to qualitatively evaluate the deviation gap of the conventional approach that ignores the side information. The accuracy of our proposed method, alongside its comparison with the simplified results of the conventional approach, is validated by experimental studies.
Our work is promising to benefit the efficient resource planning of networked control systems with limited communication resources by providing accurate ACR computation.
\end{abstract}

\section{Introduction}

\IEEEPARstart{I}{n} recent years, emerging large-scale control systems, such as collaborative manufacturing~\cite{cai2018zero}, smart power grids~\cite{mager2019feedback}, and autonomous traffic management~\cite{lv2020event}, tend to be designed in a distributed manner where \rv{the plants and the sensors} of the systems are deployed remotely and connected using a communication network. \rv{Due to the existence of stochastic noise that may potentially degrade the system performance~\cite{ling2002robust}, sufficient measurement sampling is required such that up-to-date measurements are available to facilitate the controller design~\cite{zhang2019networked}.} Nevertheless, the communication resources of the networked systems used to maintain sufficient sampling are often limited by the power restrictions of the system, such as the mobile and portable devices of which the power mainly relies on batteries. This issue suggests activating the network communication for sampling only when needed, motivating the studies on \textit{event-based schedulers} for system communication~\cite{1184824}.
The event-based schedulers present high efficiency and excellent flexibility in terms of the consumption of the communication resources~\cite{
ding2019survey, antunes2021decentralized} and have been widely used for networked systems~\cite{heemels2012introduction,han2015stochastic,mamduhi2017error}). 

\rv{A networked control system perturbed by stochastic noise and triggered by an event-based communication scheduler can be formulated as a networked event-triggered stochastic control system (NET-SCS), as} illustrated in Fig.~\ref{fig:close_loop}. \rv{Similar modeling methods are commonly seen in the literature on event-triggered systems~\cite{ling2002robust, wang2010event}.} The activation of the network communication enabling the measurement sampling is dominated by a time-asynchronous event-based \textit{scheduler}. This event does not explicitly depend on time but is associated with the variable of interest and predefined thresholds, activating the communication only when necessary. This scheme has been widely used in practical networked systems for efficient consumption of the communication resources~\cite{antunes2014rollout, demirel2016threshold, molin2014optimal, mamduhi_MTNS, liu2017event}.

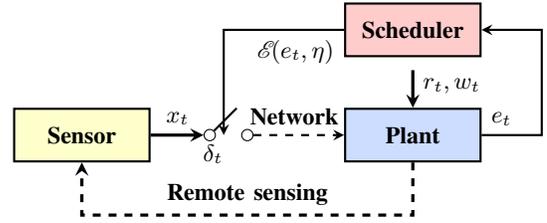
\begin{figure}[htbp]
\noindent
\hspace*{\fill} 
\begin{tikzpicture}[scale=1,font=\small]

\def\sbheight{0.8cm}
\def\sbwidth{1.6cm}
\def\dev{1.5cm}

\definecolor{closedloop}{RGB}{0, 0, 0}

\definecolor{shadowyellow}{RGB}{255, 255, 204}
\definecolor{shadowred}{RGB}{255, 204, 204}
\definecolor{shadowblue}{RGB}{204, 221, 255}

\node[minimum height=0.7cm,minimum width=1.8cm,draw,thick,fill=shadowyellow] (sen) at (-2cm,0cm) {\textbf{Sensor}};

\node[minimum height=0.7cm,minimum width=1.8cm,draw,thick,fill=shadowblue] (sys) at (2.4cm,0cm) {\textbf{Plant}};

\node[minimum height=0.7cm,minimum width=1.8cm,draw,thick,fill=shadowred] (trig) at (2.4cm,1.4cm) {\textbf{Scheduler}};

\draw[->,>=stealth, very thick, dashed] (sys.south) -- ([yshift=-0.7cm] sys.south) -- node[pos=0.5,align=center, anchor=south]{\textbf{Remote sensing}} ([yshift=-0.7cm] sen.south) -- (sen.south);

\node[circle,inner sep=0pt,minimum size=1.5mm,draw, color=closedloop] (swl) at (-0.3cm,0cm) {};
\draw[thick, color=closedloop] (swl.north east) -- ([xshift=0.3cm,yshift=0.3cm] swl.north east);
\node[circle,inner sep=0pt,minimum size=1.5mm,draw,color=closedloop] (swr) at ([xshift=0.5cm] swl) {};

\draw[->,>=stealth, very thick] (sen.east) -- node[pos=0.1,align=center, anchor=south west]{$x_t$} (swl.west);
\draw[->,>=stealth, thick, dashed,color=closedloop] (swr.east) -- node[pos=0.45,align=center, anchor=south]{\textbf{Network}} (sys.west);

\node[anchor=north east](swl.north east){$\delta_t$};

\draw[->,>=stealth,thick, color=closedloop] (sys.east) -- node[pos=0,align=center, anchor=south west]{$e_t$} ([xshift=0.8cm] sys.east) -- ([xshift=0.8cm] trig.east) -- (trig.east);

\draw[->,>=stealth,thick, color=closedloop] (trig.west) -- node[pos=0,align=center, anchor=north east]{$\mathscr{E}\!(e_t, \eta)$} ([xshift=-1.6cm] trig.west) -- ([xshift=-1.6cm, yshift=0cm] sys.west);

\draw[->,>=stealth, very thick] ([yshift=0.5cm] sys.north) -- node[pos=0.9,align=center, anchor=south west]{$r_t, w_t$} (sys.north);

\end{tikzpicture}
\hspace{\fill} 
\caption{\rv{A networked event-triggered stochastic control system (NET-SCS) consisting of a stochastic plant, a remote sensor, and an event-based scheduler, where $x_t$, $e_t$, $r_t$, and $w_t$ are the measurement, the variables of interest, control command, and stochastic noise of the system, respectively. The dashed dark-color arrow indicates remote sensing, such as visual perception, and the dashed light-color arrow denotes the networked communication between the plant and the sensor. The communication status of the network (\textit{active} or \textit{inactive}) is depicted by a binary variable $\delta_t$ and triggered by an event $\mathscr{E}\!\left( e_t, \eta \right)$ determined by a given threshold $\eta$.}}
\label{fig:close_loop}
\end{figure}

\rv{Allocating communication resources for a NET-SCS efficiently is a challenging topic, for which the performance of event-triggered network communication is important~\cite{wang2010event}. A} common index to \rv{evaluate this performance} is the average communication rate (ACR) that denotes the expectation of active communication at a certain time~\cite{ebner2016communication}, providing a practical insight into how many system resources are allocated to network communication and supports the efficient planning of system resources. Particular attention has been attracted to the \textit{stationary ACR}~\cite{molin2014price} which represents the limit of the ACR as the time approaches infinity. It is often used to evaluate the consumption of communication resources in the steady state of a networked system. Conventionally, the computation of the ACR is solved using statistical methods via numerical experiments, such as a Monte Carlo experiment. An analytical solution of the ACR with a closed form is difficult to obtain due to the \textit{side information} of event-based schedulers~\cite{8619752}. 

\rv{The \textit{side information} of a NET-SCS refers to the nontrivial relation between the triggering event and the variable of interest, usually caused by the nonlinearity of an event-based scheduler. 
For Fig.~\ref{fig:close_loop}, this means that the event $\mathscr{E}\!(e_t, \eta)$ activating the communication is nonlinear to the variable of interest $e_t$ due to deterministic thresholds. Thus, the statistics propagation of the communication status becomes nonlinear, making the ACR difficult to compute. Specifically, for deterministic thresholds-based schedulers, computing ACR involves tracking the propagation of iteratively truncated distributions, addressing a challenging problem that has not led to satisfactory results. Existing work} mainly focuses on the control and filtering of networked systems~\cite{wu2012event, wu2013can, shi2014event, ebner2016communication, han2015stochastic, demirel2017performance}, which does not require precise computation of ACR. 
In~\cite{wu2012event, ebner2016communication}, the ACR is computed based on a Gaussian distribution, without incorporating the \rv{side information}, leading to an accuracy gap with the true ACR. The work in~\cite{shi2014event} provides lower and upper bounds for the ACR without giving its analytical form. In~\cite{han2015stochastic, demirel2017performance}, the computation of ACR is simplified with a stochastic triggering threshold of which the main shortcomings are inferior control performance and sensitivity to data loss. These approaches either over-simplify the computation of the ACR by imposing impractical assumptions or conduct conservative and coarse approximations. They can hardly be used to predict exact ACR which is a fundamental step towards the study of the crucial filtering problem~\cite{shi2014event, leong2016sensor, xia2016networked, mohammadi2017event, zou2018recursive}.

\rv{This paper provides the first method to accurately compute the ACR of a NET-SCS with deterministic threshold-based triggering events. Inspired by~\cite{ebner2016communication}, we develop a recursive model to characterize the temporal evolution of ACR. Using this recursive model, we prove the existence of the stationary ACR and calculate its precise value. Moreover, we have provided analytical and numerical methods to compute ACR at an arbitrary time by precisely characterizing the statistics propagation of iteratively truncated distributions.} Theoretical analysis and experimental studies are provided to verify the accuracy of the proposed methods and qualify the inaccuracy gap of the conventional methods~\cite{wu2012event, wu2013can}. 
\rv{By providing a more precise characterization for the performance of network communication than existing methods, our work is promising to promote the performance of the existing work on computation resource allocation of networked systems~\cite{shi2014event, ebner2016communication, han2015stochastic, demirel2017performance}.} 
Our main contributions are summarized as follows.
\rv{\begin{enumerate}
    \item A novel recursive model to characterize the temporal evolution model for the ACR of a generic NET-SCS.
    \item Methods to compute the precise values of the stationary ACR and the ACR at any time.
    \item Qualitative analysis and numerical studies to address the accuracy gap of the conventional method.
\end{enumerate}}
The rest of this paper is organized as follows. Sec.~\ref{sec:prelim} formulates the problem, with mathematical preliminaries provided. In Sec.~\ref{sec:main_result_delay_free}, we introduce the recursive model of the ACR and investigate the existence of the stationary ACR. Sec.~\ref{sec:coeffi} presents the methods to compute the ACR. In Sec.~\ref{sec:compa}, we use theoretical analysis and a numerical example to verify the accuracy of the proposed method and evaluate the accuracy gap of the conventional method. In Sec.~\ref{sec:simulation}, a numerical experiment on a simple vehicle-following case is conducted to validate our methods. Finally, Sec.~\ref{sec:conclusions} concludes the paper.

\textit{Notation}: The rest of this article obeys the following notations. The sets of real and natural numbers are denoted by $\mathbb{R}$ and $\mathbb{N}$. The superscript $^+$ sued after them indicates the subsets only containing the positive elements. A Gaussian distribution with mean value $\mu \in \mathbb{R}$ and variance $\sigma^2$, $\sigma \in \mathbb{R}^+$ is represented by $\mathcal{N}(\mu,\sigma^2)$. For a stochastic event $\mathscr{E}$ defined on a probability space, $P(\mathscr{E})$ denotes the occurring probability of $\mathscr{E}$. For a stochastic variable $z \!\in \mathbb{R}$, $P(z)$, $p_z(\cdot)$, $F_z(\cdot)$, $\E(z)$, and $\mathrm{Var}(z)$ denote, respectively, its probability, (probability density function) PDF, cumulative distribution function (CDF), expectation, and variance.

\section{Problem Statement and Preliminaries} \label{sec:prelim}

In this section, we describe the main mathematical problem of this paper based on preliminaries. 

\subsection{Dynamic Model of NET-SCS and State Estimation Error}\label{sec:tdm}

A NET-SCS comprises a series of \textit{plants} and \textit{sensors} connected through a communication network. Without affecting the results in this paper, we only investigate one single plant and its associated sensor and scheduler, as shown in Fig.~\ref{fig:block_diagram}. We also assume that the dynamic model of the plant has the following linear scalar stochastic difference equation (SDE),
\begin{equation}
    x_{k+1} = Ax_{k} + Bu_{k} + w_{k}, 
    \label{eq:plant}
\end{equation}
where $k\in \mathbb{N}$ denotes the discrete sampling time of the system, $x_k, u_k\in\mathbb{R}$ are, respectively, the state and control input of the system at time $k$, $A,B\in\mathbb{R}$ are constant parameters, and $w_k\!\in\mathbb{R}$ is the stochastic noise of the system. For simplicity, we assume that the initial state of the system $x_0$ is a known deterministic variable. Note that the stochastic process $w_k$, $k \in \mathbb{N}$, is subject to the following assumption.
\begin{assumption}\label{as:as1}
The Gaussian stochastic process $w_k$ is independent and identically distributed (i.i.d.) for all $k \in \mathbb{N}$, i.e.,
\begin{enumerate}[leftmargin=*, wide, labelwidth=0pt]
    \item $w_k\sim\mathcal{N}\!\left(0,\sigma^2\right)$, $\forall \, k \in \mathbb{N}$, with $\sigma \in \mathbb{R}^+$. 
    \item $p_{w,w}(w_i,w_j) \!=\! p_w(w_i)p_w(w_j)$ holds for all $i,j \!\in\! \mathbb{N}$, $i \!\neq\! j$, where $p_w(\cdot)$ is the PDF of stochastic variable $w_k$, $k \in \mathbb{N}$, and $p_{w,w}(\cdot,\cdot)$ is the joint PDF of $w_i$ and $w_j$, $i,j \in \mathbb{N}$.
\end{enumerate}
\end{assumption}

\begin{figure}[htbp]
\noindent
\hspace*{\fill} 
\begin{tikzpicture}[scale=1,font=\scriptsize]

\def\sbheight{0.6cm}
\def\sbwidth{1.6cm}
\def\dev{1.5cm}

\definecolor{closedloop}{RGB}{0, 0, 0}

\definecolor{darkgray}{RGB}{220, 220, 220}
\definecolor{shadowgray}{RGB}{245, 245, 245}

\definecolor{shadowyellow}{RGB}{255, 255, 204}
\definecolor{shadowred}{RGB}{255, 204, 204}
\definecolor{shadowblue}{RGB}{204, 221, 255}
\definecolor{darkblue}{RGB}{31, 31, 122}
\definecolor{shadowori}{RGB}{255, 255, 229}


\node[minimum height=3.9cm,minimum width=5.5cm,draw,thick,densely dotted,fill=shadowblue] (bigblock2) at (-4.55cm,-0.1cm) {};

\node[circle,inner sep=0pt,minimum size=0.4cm,draw,fill=darkgray] (plus1) at (-7cm,0) {};
\draw[thick] (plus1.north east) -- (plus1.south west);
\draw[thick] (plus1.north west) -- (plus1.south east);

\node[minimum height=\sbheight,draw,color=closedloop,text width=\sbwidth,align=center,fill=white,rounded corners=2mm] (se) at (-5.3cm, -\dev) {\color{closedloop} State Estimator};
\node[minimum height=\sbheight,draw,text width=0.8*\sbwidth,align=center,fill=white,rounded corners=2mm] (controller) at (-5.5cm, 0) {Controller};

\node[minimum height=\sbheight,draw,text width=\sbwidth,align=center,fill=white,rounded corners=2mm] (plant) at (-3.0cm, 0) {Plant Dynamics};

\node[circle,inner sep=0pt,minimum size=0.4cm,draw,color=closedloop,thick] (plus2) at (-4cm,\dev) {};
\draw[color=closedloop,thick] (plus2.north east) -- (plus2.south west);
\draw[color=closedloop,thick] (plus2.north west) -- (plus2.south east);

\draw[->,>=stealth,thick,color=closedloop] ([yshift=0.4cm] plant.north) -- node[pos=0.7, align=left, anchor=west]{$w_t$} (plant.north);

\node[circle,inner sep=0pt,minimum size=1.5mm,draw, color=closedloop] (swl) at (0cm,-0.5cm) {};
\draw[thick, color=closedloop] (swl.south east) -- ([xshift=0.3cm,yshift=-0.3cm] swl.south east);
\node[circle,inner sep=0pt,minimum size=1.5mm,draw, color=closedloop] (swr) at ([yshift=-0.5cm] swl) {};

\node[minimum height=\sbheight,draw,color=closedloop,text width=0.8*\sbwidth,align=center,fill=shadowred, rounded corners=2mm] (robot) at (-0.7cm, \dev) {\textbf{Scheduler}};

\node[minimum height=\sbheight,draw,color=closedloop,text width=0.45*\sbwidth,align=center,fill=shadowyellow, rounded corners=2mm] (sen) at (-0.7cm, 0cm) {\textbf{Sensor}};

\node[minimum height=\sbheight,draw,color=closedloop,text width=0.8*\sbwidth,align=center,fill=white, rounded corners=2mm] (mem) at (-2.8cm, -\dev) {\color{closedloop} Memory};

\draw[->,>=stealth,thick] ([xshift=-0.6cm] plus1.west) -- node[pos=0.8,align=left, anchor=south east]{$r_k$} (plus1.west);

\draw[->,>=stealth,thick] (plus1.east) -- (controller.west);

\draw[->,>=stealth,thick] (controller.east) -- node[pos=0,align=left, anchor=south west]{$u_k$} (plant.west);

\draw[->,>=stealth,color=closedloop,thick] (plus2.east) -- node[pos=0.82,align=right, anchor=south west]{$e_k$} (robot.west);

\draw[->,>=stealth,thick] (se.west) -- ([xshift=-0.1cm] se.west) -- ([xshift=-1.7cm] se.center) --node[pos=1,anchor=north east]{$-$} (plus1.south);

\draw[->,>=stealth,thick,densely dashed,color=closedloop] (swr.south) --  ([yshift=-\dev+1cm] swr.center) -- node[pos=0.3,align=center, anchor=north]{\textbf{Network}} (mem.east);

\draw[->,>=stealth,thick,color=closedloop] (mem.west) -- node[align=left, anchor=south]{$\mathcal{I}_{k-1:\iota}$} (se.east);

\node[align=left, anchor=north west] at (bigblock2.north west) {\textbf{Plant}};

\draw[thick] (plant.east) -- ([xshift=0.15cm] plant.east);

\draw[->,>=stealth,thick,dashed] (plant.east) -- node[pos=0.5,align=center, anchor=north,text width=1.2cm]{\textbf{Remote} \\ \textbf{sensing}}(sen.west);

\draw[->,>=stealth,thick] (sen.east) -- ([yshift=0.5cm] swl.center) -- node[pos=1,align=right, anchor=north east]{$\delta_k$} (swl.north);

\node[anchor=south west] at (sen.east){$x_k$};

\draw[->,>=stealth,thick] ([xshift=0.15cm] plant.east) --  ([xshift=0.15cm, yshift=0.8cm] plant.east) -- 
([yshift=-\dev+0.8cm] plus2.center) --  (plus2.south);
\node[circle,inner sep=0pt,minimum size=1mm,draw,fill=black] at ([xshift=0.15cm] plant.east) {};

\node[align=right, anchor=south east,color=closedloop] at ([xshift=-1mm] plus2.center) {$-$};

\draw[->,>=stealth,color=closedloop,thick] (robot.east) -- node[pos=0,align=center, anchor=south west]{$\mathscr{E}_k$} ([xshift=0.5cm,yshift=\dev+0.5cm] swl.center) -- ([xshift=0.5cm,yshift=-0.25cm] swl.center) -- ([xshift=0.2cm,yshift=-0.25cm] swl.center);

\draw[->,>=stealth,thick] ([xshift=1.1cm] controller.center) -- ([xshift=1.1cm,yshift=-0.8cm] controller.center) -- ([yshift=0.4cm] mem.north) --(mem.north);

\node[circle,inner sep=0pt,minimum size=1mm,draw,fill=black] at ([xshift=1.1cm] controller.center) {};

\draw[->,>=stealth,thick,color=closedloop] ([xshift=-0.3cm] se.west) -- node[pos=0,align=center, anchor=north]{$\hat{x}_k$} ([xshift=-0.3cm, yshift=3cm] se.west) --  (plus2.west);
\node[circle,inner sep=0pt,minimum size=1mm,draw,fill=closedloop] at ([xshift=-0.3cm] se.west) {};

\end{tikzpicture}
\hspace{\fill} 
\caption{The block diagram of a NET-SCS.}
\label{fig:block_diagram}
\end{figure}
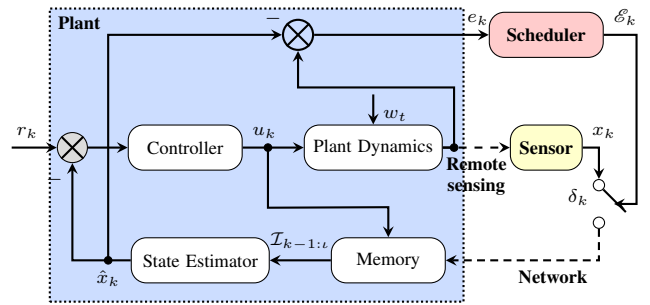

\begin{remark}
Our work in this paper only deals with a scalar NET-SCS. Note that the exact computation of the ACR for a multi-dimensional NET-SCS is even more challenging due to the probabilistic coupling among the individual dimensions of the non-Gaussian system variables. Also, various triggering options for multi-dimensional system variables make it tedious to determine their analytical distribution functions. The study on a scalar NET-SCS is sufficient to draw essential quantitative and qualitative conclusions that can be extended to multi-dimensional systems with additional efforts in future work.
\end{remark}

For system \eqref{eq:plant}, its state $x_k$ at any time $k \in \mathbb{N}^+$ can be measured by one or multiple sensors. However, the measurement is only sampled when the network communication is active, indicated by a closed switch in Fig.~\ref{fig:block_diagram}. At the time $k$, whether the status of the communication is active or not is represented as a binary event $\delta_k\in \{0,1\}$, namely $\delta_k = 1$ for active and $\delta_k = 0$ for inactive. For brevity, we represent these conditions as $\delta_k^{\{1\}}$ and $\delta_k^{\{0\}}$. The communication status $\delta_k$ is affected by an event $\mathscr{E}_k$ which is produced by an event-based \textit{scheduler} which will be introduced in Sec \ref{sec:trigger}. When the system state is not sampled, an estimated value $\hat{x}_k$ is provided by a \textit{state estimator} utilizing the system model \eqref{eq:plant} and the history data $\mathcal{I}_{k-1:\iota}= \left\{ \delta_{k-1},\cdots,\delta_{\iota+1}, \delta_{\iota},x_{\iota}, u_{k-1},\cdots,u_{\iota} \right\}$~\cite{mamduhi2017error}, 
\begin{equation}\label{eq:est}
\begin{split}
\hat{x}_{\iota} = &\, x_{\iota} \\
\hat{x}_{\iota+1} = & \, A x_{\iota} + B u_{\iota}, \\
\cdots &\,  \\
\hat{x}_{i} =& \, A \hat{x}_{i-1} + B u_{i-1},
\end{split}
\end{equation}
$i = \iota+2, \cdots, k$, where $\iota \in \mathbb{N}$ refers to the last instant of state sampling and equation $\hat{x}_{\iota} =  x_{\iota}$ indicates a state sampling operation. Since the initial state $x_0$ is deterministic and known, we have $\hat{x}_0 = x_0$ and $e_0=0$. 
All the system data before the last instant, i.e., all $\{\delta_i, x_i, u_i\}$ where $i < \iota$, are timely abandoned to limit the scale of the \textit{memory} used to store the data. We make a correspondence between the NET-SCS in Fig.~\ref{fig:block_diagram} and the general event-based networked system in Fig.~\ref{fig:close_loop} by recognizing $x_k$ as the system state and the estimation error $e_k = x_k - \hat{x}_k$ as the system variable of interest.

\rv{
\begin{remark}
We assume the control input $u_k$ in Eq.~\eqref{eq:plant} designed as {\color{black}a \textit{feedback} controller $u_k =u(\hat{x}_k-r_k)$ \rv{based on the state estimation $\hat{x}_k$} to achieve the desired closed-loop performance.} Incorporating the dynamics of state measurements in Eq.~\eqref{eq:est}, the control input $u_k$ at any time $k$ is deterministic and does not influence the analysis of ACR, as shown later. However, it is still necessary to include $u_k$ in Eq.~\eqref{eq:plant} to imply how a NET-SCS is controlled, although {\color{black}the performance of the feedback controller is not within the scope of this paper. } 
\end{remark}
}

By subtracting (\ref{eq:est}) from (\ref{eq:plant}), we obtain the dynamic model of the state estimation error as
\begin{equation}\label{eq:open}
\begin{split}
e_{\iota} = &~ 0,\\
e_{\iota+1} =&~ w_{\iota}, \\
\cdots & \\
e_{i} =& ~ A e_{i-1} + w_{i-1}, 
\end{split}
\end{equation}
where $i = \iota+2, \cdots, k$. Therefore, for all $i = \iota+1, \cdots, k$, $e_i$ is a random variable. Substituting $e_{k-1}$, $e_{k-2}$, $\cdots$, $e_{\iota}$ to $e_k$ recursively, we obtain
\begin{equation}\label{eq:open1}
\textstyle e_{k} = \sum_{i=\iota}^{k-1} A^{k-i-1}w_i,~k>\iota.
\end{equation}

\rv{
\begin{remark}
Eq.~\eqref{eq:open1} is based on assuming a deterministic and known initial state $x_0$. If $x_0$ is random and unknown, one can take its expected value as its initial estimation, i.e., $\hat{x}_0 \!=\! \E(x_0)$. This leads to a drifted estimation error $\textstyle e_{k} \!=\! \sum_{i=\iota}^{k-1} A^{k-i-1}w_i \!+\! A(x_0 \!-\! \E(x_0))$ for any $k$ between $0$ and the first measurement sampling time. However, after the first measurement sampling, the estimation error $e_k$ is the same as Eq.~\eqref{eq:open1} for all $k\in \mathbb{N^+}$. Analyzing the drifted estimation error particularly before the first measurement sampling is not challenging but tedious. Therefore, we only assume known and deterministic initial states without losing generality. 
\end{remark}
}

Indicated by Eq.~\eqref{eq:open1}, the estimation error accumulates over time $\{\iota+1, \iota+2, \cdots, k\}$ due to inactive communication. This may lead to possible performance degradation of the system performance. Since the accumulated error $e_k$ is a linear combination of Gaussian variables, it is subject to a Gaussian distribution $\mathcal{N}\!\left(0, \rv{\sigma^2}\sum_{i=\iota}^{k-1}\! A^{2(k-i-1)} \right)$. However, when network communication is triggered by deterministic threshold-based events, $e_k$ does not follow a Gaussian distribution. We will discuss this case in the following subsection.


\subsection{Event-Triggered Scheduler with Side Information}\label{sec:trigger}

\rv{On the one hand, network communication should be activated as infrequently as possible to reduce the consumption of communication resources. On the other hand, the duration of deactivated communication should not be arbitrarily long. It is necessary to activate the communication when} the last state estimation error $e_{k-1}$ exceeds a predefined threshold $\eta \in \mathbb{R}^+$ \rv{to avoid} the accumulation of the state estimation error $e_k$. \rv{Additionally, to improve the robustness of the system to faults and anomalies,} the communication should be activated when the consecutive inactive period is beyond a time limit $T \in \mathbb{N}^+$, \rv{even when the error $e_k$ does not exceed the thresholds. Time limits for inactive communication periods have been set in existing work~\cite{heemels2010networked, dolk2016event}.
These considerations render the following event-triggered scheduler,}
\begin{equation}\label{eq:schedul_variables_policy}
   \delta_k=\begin{cases} 1, & \qquad \text{if} \quad |e_{k-1}|\geq \eta ~\text{or}~ k-\iota > T  \\ 0, & \qquad \text{otherwise}, \end{cases}
\end{equation}
for any $k\in \mathbb{N}^+$ and $\iota \in \mathbb{N}$, $\iota < k$, where the condition that triggers active communication $\delta_k^{\{1\}}$ refers to a positive event $\mathscr{E}_k$, otherwise a negative event $\overline{\mathscr{E}}_k$. The scheduling scheme \eqref{eq:schedul_variables_policy} maintains a decent system performance by ensuring low resource usage by limiting the frequency of communication while restricting the estimation errors. Such an event-based scheduling model is widely used in various networked control systems, such as platooning of a group of vehicles \cite{dolk2017event, liu2017joint}, power systems \cite{dong2016event, saxena2020event} and cooperative manipulation in robotics systems \cite{dohmann2020distributed, ngo2021event}. 

The deterministic threshold-based event-triggering scheduling scheme in Eq.~\eqref{eq:schedul_variables_policy} brings saturation nonlinearity to the statistical propagation $\delta_k$ as $k$ increases. Specifically, the error accumulation depicted in~\eqref{eq:open} is only valid when $|e_i| \leq \eta$, for any $i = \iota, \iota+1,\cdots,k-1$. Otherwise, any $|e_i|>\eta$ will immediately activate the state sampling and lead to $\delta_{i+1} = 1$ and $e_{i+1} = 0$.
Thus, given the last communication instant $\iota$ and the history data $\mathcal{I}_{k-1:\iota}$, the estimation error under the event-triggered scheduling scheme (\ref{eq:schedul_variables_policy}) becomes
\begin{equation}\label{eq:close}
\begin{split}
\hat{e}_{\iota} =&~ 0, \\
\hat{e}_{\iota+1} =&~ w_{\iota}, \\
\cdots & \\
\hat{e}_{i} =& \left\{ \begin{array}{ll}
A \hat{e}_{i-1} + w_{i-1}, ~&\mathrm{if} \, |\hat{e}_{i-1}|< \eta, \\
0,~& \mathrm{else},
\end{array} \right. 
\end{split}
\end{equation}
where $i=\iota+2, \cdots, k$, for all $k \leq \iota+T$. Here, we use a new symbol $\hat{e}_i$ to represent the estimation error with deterministic threshold-based triggering Eq.~\eqref{eq:schedul_variables_policy}, distinguished from the original estimation error $e_i$ in \eqref{eq:open} without event-based triggering. Similar to $e_i$, $\hat{e}_i$ is also a random variable, for $i = \iota+1, \iota+2, \cdots, k$. Differently, $\hat{e}_i$ assigned in \eqref{eq:close} incorporates \textit{side information} $|\hat{e}_{i-1}|< \eta$ compared to $e_i$ \eqref{eq:open}. This side information imposes a truncation operation on the probabilistic distribution of the estimation error at every sampling instant, leading to nonlinearity in its statistical propagation. As a result, the estimation error is hardly subject to a Gaussian distribution as time increases, even though the noise is a stationary Gaussian process according to Assumption \ref{as:as1}. Also, it is difficult to bring up a brief overall analytical form to represent $\hat{e}_k$, similar to \eqref{eq:open1}, which makes it difficult to track the nonlinear statistical propagation. Here, we refer to the nonlinearity brought to the statistical propagation of the system variable of interest as the \textit{side information}.

\subsection{Stationary and Transient ACR of a NET-SCS}  \label{sec:tacr}
For the NET-SCS in (\ref{eq:plant}) with the state estimator (\ref{eq:est}) and the event-triggered scheduler (\ref{eq:schedul_variables_policy}), the ACR is defined as
\begin{equation}
\E\!\left(\delta_k\right) = P(\delta_k^{\{1\}} ),~k \in \mathbb{N},
\label{eq:expectation_delta_k}
\end{equation}
which describes the likelihood of activating the communication or the probability of the state sampling. Since the initial state $x_0$ is deterministic and known, we have $\E(\delta_0) $$=$$ P(\delta_0^{\{1\}} ) $$= 1$. Also, given $\hat{e}_0 = x_0 - \hat{x}_0 =0$, we know $\E(\delta_1) $$= $$P(\delta_1^{\{1\}}) $$=$$ 0$. For any $k\in \mathbb{N}^+$, $k \geq 2$, however, the value of $\delta_k$ is usually random, and the ACR is computed as
\begin{equation}
\textstyle \E\!\left(\delta_k\right) = 1- P(\delta_k^{\{0\}} ) = 1 - \int_{-\eta}^{\eta} p_{\hat{e}_{k-1}}(z) \mathrm{d} z,
\label{eq:expectation_delta_k_2}
\end{equation}
where $p_{\hat{e}_{k}}(\cdot)$ is the PDF of the state estimation error and $z \in \mathbb{R}$ is an auxiliary variable. Note that the integration interval $[\,-\eta, \eta\,]$ corresponds to the constraint $|\hat{e}_{k-1}| < \eta$ in (\ref{eq:schedul_variables_policy}) which blocks state sampling at time $k$. 

\rv{
\begin{remark}\label{eq:ACRpar}
From Eq.~\eqref{eq:expectation_delta_k_2}, it can be inferred that the value of ACR $\E(\delta_k)$ is affected by the threshold $\eta$. Specifically, a smaller $\eta$ may lead to a larger $\E(\delta_k)$. Additionally, $\E(\delta_k)$ may also increase with a larger system parameter $A$ in Eq.~\eqref{eq:plant} since the estimation error $e_k$ may have a bigger variance, making it easier to exceed the threshold $\eta$.
\end{remark}
}

\begin{remark}
The above statements are based on our assumption that the initial system state $x_0$ is known. Otherwise, the values of $E(\delta_0)$ and $E(\delta_1)$ are neither $1$ nor $0$. Instead, they are dependent on the distribution of $x_0$ and should be valued within the interval $(0,\,1)$.
\end{remark}

The limit $\E(\delta_{\infty}) = \displaystyle \lim_{k \rightarrow \infty} \E(\delta_k)$, if it exists, is defined as the \textit{stationary ACR}, while $\E(\delta_k)$ for a finite $k \in \mathbb{N}$ is referred to as the transient ACR. The main issue of exactly calculating the stationary and the transient ACR is that the analytical form of the PDF $p_{\hat{e}_k}(\cdot)$, $k\in\mathbb{N}$, is difficult to track due to the challenge of capturing the nontrivial statistical propagation of the estimation errors, as introduced in Sec.~\ref{sec:trigger}. 
In some existing work~\cite{ebner2016communication, demirel2017performance}, $p_{\hat{e}_k}(\cdot)$ is approximated by a Gaussian distribution without incorporating the side information, which leads to the approximated calculation of ACR. This article shows that such an approximation results in a larger ACR than the ground truth. Alongside this, accurate ACR computation methods are also provided.



\section{Communication rate analysis}\label{sec:main_result_delay_free}

This section analyzes the transient and the stationary ACR of a NET-SCS. A recursive model for the transient ACR is derived to prove the existence of the stationary ACR. This model allows calculating the transient and 
stationary ACR values using a finite number of coefficients.

\subsection{The Predictive Indexes and The Predictive Coefficients}\label{sec:rec_pic}

In this section, we introduce the predictive indexes and the predictive coefficients which are important to analyze the ACR. We first define a compound event for $k,n\in \mathbb{N}^+$, $n \leq k$,
\begin{equation*}
\mathscr{E}_{k:k-n} = \delta_{k}^{\{0\}} \cap \delta_{k-1}^{\{0\}} \cap \cdots \cap \delta_{k-n+1}^{\{0\}} \cap \delta_{k-n}^{\{1\}},
\end{equation*}
which represents the conjunction of $n$ successive inactive events after an active event $\delta_{k-n}^{\{1\}}$.
It is straightforward to show that the compound event satisfies the following property.
\begin{property}\label{py:prob}
For any $n,k  \in \mathbb{N}^+$, $n \leq k$, event $\mathscr{E}_{k:k-n}$ satisfies the following conditions.
\begin{enumerate}
\item \rv{$P(\mathscr{E}_{k:k-n}) > 0$}, $\forall \, n \leq T$, and \rv{$P(\mathscr{E}_{k:k-n}) = 0$}, $\forall \, n >T$.
\item \rv{$P(\mathscr{E}_{k:k-i} \cap \mathscr{E}_{k:k-j}) = 0$}, for any $i,j \leq k$, $i \neq j$.
\item $\bigcup_{n=1}^k \mathscr{E}_{k:k-n} = \delta_{k}^{\{0\}}$. 
\end{enumerate}
\end{property}
In Property \ref{py:prob}, condition 1) is met considering that the non-communication period of the system should not be larger than the limit $T$, according to the event-triggered scheduler~\eqref{eq:schedul_variables_policy}. Condition 2) is verified by the mutual exclusion between the compound events. Condition 3) is justified by taking the union of all compound events $\mathscr{E}_{k:k-n}$, for all $n=1,2,\cdots,k$. 

\subsubsection{Predictive indexes}
For particular communication variables $\delta_0$, $\delta_1$, $\cdots$, $\delta_n$, we define $n$-\textit{step predictive non-communication index} $P_n$, or \textit{predictive index} as
\begin{equation}\label{eq:n_index}
P_n = P\!\left( \left. \delta_{n}^{\{0\}}, \delta_{n-1}^{\{0\}}, \cdots, \delta_{1}^{\{0\}} \right| \delta_{0}^{\{1\}} \right),~n \in \mathbb{N}^+,
\end{equation}
which denotes the probability that no communication is activated for $n$ sampling instants given an active communication event $\delta_0^{\{1\}}$. We will later generalize the predictive index to arbitrary-time communication status $\delta_k$, $k \in \mathbb{N}^+$. The $n$-step predictive index $P_n$ satisfies the following property.

\begin{property}\label{pr:sum}
For any $n, T \in \mathbb{N}^+$, the predictive index $P_{n}$ satisfies the following conditions.
\begin{enumerate}
\item For all $n > T$, $P_{n} =0$.
\label{pp:en:1}
\item For all $n \leq T$, $0< P_{n} < 1$.
\label{pp:en:2}
\end{enumerate}
\end{property}
Property \ref{pr:sum}-\ref{pp:en:1}) is justified by that the communication is activated by force after $n>T$, according to \eqref{eq:schedul_variables_policy}. For $n \leq T$, neither activation nor deactivation of the communication is a certain event, since the support of the PDF of the noise $w_k$ is infinite, $\forall \, k = 1,2, \cdots,n$. This addresses property \ref{pr:sum}-\ref{pp:en:2}). 

\subsubsection{Predictive coefficients}
Also for variables $\delta_0$, $\delta_1$, $\cdots$, $\delta_n$, the $n$-\textit{stacked predictive non-communication coefficient} $\overline{P}_n$, or \textit{predictive coefficient} is defined as
\begin{equation}\label{eq:n_coef}
\overline{P}_n = P \!\left( \left. \delta_n^{\{0\}} \right| \mathscr{E}_{n-1:0} \right),~n=1,2,\cdots, T,
\end{equation}
which denotes the probability of a single non-communication event $\delta_n^{\{0\}}$ given the history compound event $\mathscr{E}_{n-1:0}$. 
Similar to the predictive index, the predictive coefficient has the following property due to the infinite support of the PDF of the stochastic noise.

\begin{property}\label{pp:overP}
$0 < \overline{P}_n < 1$ holds for all $n = 1, 2, \cdots,T$.
\end{property}

\subsubsection{Relation between the indexes and coefficients}
From the definitions of the predictive index in~\eqref{eq:n_index} and the predictive coefficient~\eqref{eq:n_coef}, we have the following relation,
\begin{equation}\label{eq:pp2}
\begin{split}
P_n =& \, P\!\left( \left. \delta_{n}^{\{0\}}, \delta_{n-1}^{\{0\}}, \cdots, \delta_{1}^{\{0\}} \right| \delta_{0}^{\{1\}}  \right) \\
=& \, P\!\left( \left. \delta_{n}^{\{0\}} \right| \mathscr{E}_{n-1:0}  \right)  P\!\left( \left. \delta_{n-1}^{\{0\}}, \cdots, \delta_{1}^{\{0\}} \right| \delta_{0}^{\{1\}}  \right) \\
=& \, \overline{P}_n P_{n-1},~n=2,\cdots,T,
\end{split}
\end{equation}
with $P_1 = \overline{P}_1$, which renders the following property.

\begin{property}\label{pp:rel}
(Relation of the predictive index and coefficient) The following relation holds.
\begin{equation}\label{eq:PPI}
\textstyle P_n =
\prod_{i=1}^n \overline{P}_i,~n = 1,2,\cdots, T.
\end{equation}
\end{property}

Meanwhile, applying Property~\ref{pr:sum} and Property~\ref{pp:overP} to \eqref{eq:pp2} recursively, we have the following property.
\begin{property}\label{pp:coef}
(Monotonic decrease of predictive index) The condition
$ 0 < P_T < \ldots < P_2 < P_1 < 1$ always holds.
\label{lem:transtion_inequality}
\end{property}

\subsubsection{Time-Invariance of the indexes and coefficients}

Assumption \ref{as:as1} ensures the stochastic process $w_k$, $k \in \mathbb{N}^+$ to be stationary. Thus, its stochastic properties are time-invariant. Therefore, the distribution of the state estimation error in~\eqref{eq:close} and the communication status in~\eqref{eq:schedul_variables_policy} is also invariant to the last communication instant $\iota$. This justifies the following property.

\begin{property}\label{pp:inva}
(Time-invariance of communication probabilities) The following conditions hold $\forall \, n =1,2,\cdots,T$ and $\forall \, \iota \in \mathbb{N}$.
\begin{equation}\label{eq:ppbil}
\begin{split}
P\!\left( \left. \delta_{\iota+n}^{\{0\}}, \delta_{\iota+n-1}^{\{0\}}, \cdots, \delta_{\iota+1}^{\{0\}} \right| \delta_{\iota}^{\{1\}} \right)\, &= P_n, \\
P \!\left( \left. \delta_{\iota+n}^{\{0\}} \right| \mathscr{E}_{\iota+n-1:\iota} \right)\, &= \overline{P}_n.
\end{split}
\end{equation}
\end{property}

Property \ref{pp:inva} proposes an important claim that the predictive indexes $P_n$ and coefficients $\overline{P}_n$ can be used to depict the communication probabilities \eqref{eq:ppbil} for an arbitrary state-sampling time $\iota \in \mathbb{N}$, even though they are originally defined specifically for $\iota = 0$. Nevertheless, we should keep in mind that this property only holds when Assumption \ref{as:as1} is ensured.

\subsection{The Recursive Model of The Transient ACR}\label{sec:rec_communication_rate}

Having introduced the predictive indexes and coefficients, we are ready to present the recursive model for the transient ACR. 
According to Property \ref{py:prob}-3) and \ref{py:prob}-2), we know
\begin{equation*}
P\!\left(\delta_{k}^{\{0\}} \right) = P\!\left( \bigcup_{n=1}^{k} \mathscr{E}_{k:k-n} \right) = \sum_{n=1}^{k} P\!\left( \mathscr{E}_{k:k-n} \right),~ k \in \mathbb{N}^+,
\end{equation*}
which leads to
\begin{equation}\label{eq:pp2r}
P \!\left(\delta_{k}^{\{0\}} \right) = \sum_{n=1}^{k} \underbrace{P\!\left( \left. \delta_{k}^{\{0\}}, \cdots, \delta_{k-n+1}^{\{0\}} \right| \delta_{k-n}^{\{1\}} \right)}_{=P_n} P\!\left( \delta_{k-n}^{\{1\}} \right),
\end{equation}
where we used Property~\ref{pp:inva}. Note that $P_n=0$ holds $\forall \, n >T$ according to Property \ref{py:prob}-\ref{pp:en:1}). Thus, \eqref{eq:pp2r} can be rewritten as
\begin{equation}\label{eq:equap}
\textstyle P \!\left(\delta_{k}^{\{0\}} \right) = \sum_{n=1}^{\min(k,T)} P_n P\!\left( \delta_{k-n}^{\{1\}} \right),~k \in \mathbb{N}^+.
\end{equation}
According to the definition of ACR in (\ref{eq:expectation_delta_k}), \eqref{eq:equap} leads to the following recursive model,
\begin{equation}\label{eq:rec}
\textstyle \E\left(\delta_k \right) = 
1 - \sum^{\min(k,T)}_{n=1} P_{n} \E\left(\delta_{k-n} \right),~k\in \mathbb{N}^+,
\end{equation}
with an initial condition $\E(\delta_0) = 1$. Model~\eqref{eq:rec} depicts the recursive evolution of ACR at an arbitrary sampling instant as time increases. Using Property~\ref{pp:rel}, (\ref{eq:rec}) can be rewritten as
\begin{equation}\label{eq:expcon}
\textstyle \E\left(\delta_k \right) = 
1 - \sum^{\min(k,T)}_{n=1} \prod_{i=1}^n \overline{P}_{i} \E\left(\delta_{k-n} \right),~k\in \mathbb{N}^+.
\end{equation}  

\rv{
\begin{remark}
Here, we should highlight that the coefficients $\overline{P}_1$, $\overline{P}_2$, $\cdots$, $\overline{P}_T$ are dependent on the threshold $\eta$ and the system parameter $A$. Therefore, Eq.~\eqref{eq:expcon} implies the influence of these parameters on the ACR $\E(\delta_k)$, even though they do not explicitly appear in Eq.~\eqref{eq:expcon}. In Remark~\ref{eq:ACRpar}, we have qualitatively analyzed the influence of $\eta$ and $A$ on ACR. We will give the quantitative relation between ACR and these parameters in Sec.~\ref{sec:coeffi}.
\end{remark}
}

The recursive model \eqref{eq:expcon} indicates that the ACR $\E(\delta_k)$ at any time $k \in \mathbb{N}^+$ can be recursively calculated using a finite number of predictive coefficients $\overline{P}_1, \overline{P}_2, \cdots, \overline{P}_T$. \rv{Moreover, we can use \eqref{eq:expcon} to obtain the stationary ACR $\E(\delta_{\infty})$ by taking its limit, if this limit exists.} A critical technical point to precisely compute the transient and stationary ACR values is to obtain the values of these coefficients, which will be explored in Sec.~\ref{sec:coeffi}. \rv{In the following subsections, we will prove the existence of the stationary ACR for a NET-SCS and give the method to calculate its value.}

\subsection{The Existence of The Stationary ACR} \label{sec:conv_ana}

The recursive model~\eqref{eq:rec}, for any $k \geq T$ with $T \in \mathbb{N}^+$, is equivalent to a $T$-order dt-LTI system. This offers us a solution to study the existence of the stationary ACR using the well-known Jury Criterion (See Sec. Appx.-A for details), rendering the following theorem.

\begin{theorem}\label{th:theorem1}
The stationary ACR $\E(\delta_{\infty}) = \displaystyle \lim_{k \rightarrow \infty} \E(\delta_{k})$ derived from the recursive model~\eqref{eq:rec} exists and its value reads
\begin{equation}
\E(\delta_{\infty}) = \textstyle {1}/({1+\sum_{n=1}^{T}\prod_{i=1}^n \overline{P}_{i}}).
\label{eq:stationary_rate}
\end{equation}
\end{theorem}

\begin{proof}
See Appx.-B. for the detailed proof.
\end{proof}

Based on a general recursive model~\eqref{eq:rec}, Theorem~\ref{th:theorem1} proves the existence of the stationary ACR for any NET-SCS with an event-triggered communication scheduler and a deterministic constant threshold. This claim does not require any additional conditions, meaning that the stationary ACR in general exists for any NET-SCS defined in this paper. Equation \eqref{eq:stationary_rate} indicates that the stationary ACR can also be calculated using a finite number of predictive coefficients $\overline{P}_n$, $n = 1,2,\cdots,T$, similar to the transient ACR explained in Sec.~\ref{sec:rec_communication_rate}.

\section{Computation of the Predictive Coefficients}\label{sec:coeffi}

As shown in Sec.~\ref{sec:main_result_delay_free}, the computation of both the transient and the stationary ACR requires the predictive coefficients $\overline{P}_n$, $n=1,2,\cdots,T$. This section provides both analytical and numerical approaches to compute these coefficients. Then, we compare our methods with the previous results which apply the restricted Gaussianity assumption. Finally, we present a numerical example to demonstrate our theoretical claims.

\subsection{The Analytical Form of The Predictive Coefficients}\label{sec:crc}

This section explores the analytical method to exactly compute the predictive coefficients.    
According to the definition of the predictive coefficients in Sec.~\ref{sec:rec_pic}, 
we have
\begin{equation}\label{eq:con2}
\textstyle \overline{P}_i = P\!\left(\left.\delta_i^{\{0\}} \right|\mathscr{E}_{i-1:0} \right) 
= \int_{-\eta}^{\eta} p_{\hat{e}_{i-1}}\!\left(z\right) \mathrm{d} z,
\end{equation}
for $i=2,\cdots,T$, with an initial condition $\overline{P}_1 = 1$.
Thus, each coefficient $\overline{P}_i$ is the integration of the PDF of the state estimation error $\hat{e}_{i-1}$ on a finite support set $[\,-\eta,\eta\,]$, where the error recursively evolves following~(\ref{eq:close}). 
Then, the critical technical point is to obtain the analytical form of these PDFs. For each $i = 1,2,\cdots,T-1$, the PDF of $\hat{e}_i$ reads
\begin{equation}\label{eq:pizi11}
\textstyle p_{\hat{e}_i}(z) = \int_{-\eta}^{\eta}p_{\hat{e}^{\eta}_{i-1}}(\xi)p_w(z - A \xi) \mathrm{d} \xi,
\end{equation}
where $p_{\hat{e}^{\eta}_{i-1}}(\cdot)$ denotes the PDF of the truncated stochastic variable $\hat{e}^{\eta}_{i-1}$ of $\hat{e}_{i-1}$ with a symmetric truncation interval $[\,-\eta,\eta\,]$ and $p_w(\cdot)$ is the PDF of the disturbance $w_k$, $k \in \mathbb{N}^+$. Note that $\hat{e}_0 = 0$, and $\hat{e}_1, w_k \sim \mathcal{N}(0,\sigma)$, which yields 
\begin{equation}\label{eq:pdf_0_1}
p_{\hat{e}_0}(z) \!=\! \delta(z),~
p_{\hat{e}_1}(z) \!=\! p_w(z)\!=\! \frac{1}{\sqrt{2\pi}\sigma} \mathrm{exp} \left( - \frac{z^2}{2\sigma^2} \right),
\end{equation}
where $\delta(\cdot)$ is the Dirac delta function.
Thus, the distribution $p_{\hat{e}_i}(z)$ in (\ref{eq:pizi11}) can be obtained as
\begin{equation}\label{eq:pizi22}
p_{\hat{e}_i}(z) = \int_{-\eta}^{\eta}\frac{p_{\hat{e}^{\eta}_{i-1}}(\xi)}{\sqrt{2\pi}\sigma} \mathrm{exp} \left(-\frac{(z-A \xi)^2}{2\sigma^2}\right) \mathrm{d} \xi.
\end{equation}
Note that Eq.~\eqref{eq:pizi22} gives a truncated Gaussian PDF. The details of truncated PDFs can be referred to in Appx.-C which renders
\begin{equation}\label{eq:ptr}
p_{\hat{e}^{\eta}_{i-1}}(z) = G_{\hat{e}_{i-1}}(z) p_{\hat{e}_{i-1}}(z),
\end{equation}
where $p_{\hat{e}_{i-1}}(\cdot)$ is the PDF of the non-truncated variable $\hat{e}_{i-1}$ and $G_{\hat{e}_{i-1}}(\cdot)$ is a piece-wise constant function defined as
\begin{equation}\label{eq:G}
G_{\hat{e}_{i-1}}(z) = \left\{ \begin{array}{ll}
1 / \int_{-\eta}^{\eta} p_{\hat{e}_{i-1}}(\xi) \mathrm{d} \xi , & -\eta \leq z \leq \eta, \\
0, & \mathrm{otherwise}.
\end{array} \right. 
\end{equation}
Substituting (\ref{eq:G}) and (\ref{eq:ptr}) to the PDF~(\ref{eq:pizi22}), we obtain
\begin{equation}\label{eq:detui1}
\!p_{\hat{e}_i}(z) =  \frac{G_{\hat{e}_{i-1}}(z)}{ \sqrt{2\pi}\sigma } \int^{\eta}_{-\eta} p_{\hat{e}_{i-1}}(\xi) \;\mathrm{exp}\!\left( \!- \frac{\left(z - A \xi \right)^2}{2\sigma^2} \right) \mathrm{d} \xi .
\end{equation}
Thus, equations \eqref{eq:pdf_0_1} and \eqref{eq:detui1} form a complete recursive model to solve the analytical forms of the PDFs of the estimation errors, $p_{\hat{e}_i}(\cdot)$, for all $i = 1, 2, \cdots, T$. Then, \eqref{eq:con2} can be used to accurately calculate the predictive coefficients $\overline{P}_i$, for $i = 2, 3, \cdots, T-1$. Note that, for any $i = 2,3,\cdots, T$, the PDF $p_{{\hat{e}}_i}(\cdot)$ is not necessarily Gaussian due to the recursive truncation operations. Also, the analytical form of $p_{{\hat{e}}_i}(\cdot)$ becomes increasingly complicated and challenging to solve as $i$ gets larger. To resolve this issue, in the next section, we propose a numerical algorithm to approximate the predictive coefficients using the recursive stochastic sampling technique.

\subsection{Approximating The Predictive Coefficients Numerically}\label{sec:nm}

Considering the difficulty of analytically computing the coefficients $\overline{P}_i$ for large $i$, we propose a numerical algorithm to approximate them using the recursive stochastic sampling method, as shown in Algorithm~\ref{ag:rme}. The computation of $\overline{P}_1$ and $\overline{P}_2$ in Line 1 is straightforward since the analytical forms of $p_{\hat{e}_0}(\cdot)$ and $p_{\hat{e}_1}(\cdot)$ are trivial and simple. In Line 2, $N$ particles are initialized from the distribution $p_{\hat{e}_1}(\cdot)$, i.e., a Gaussian distribution $\mathcal{N}(0, \sigma)$. From Line 3, the particles are used to approximate the nontrivial PDFs $p_{\hat{e}_i}(\cdot)$ for $i \geq 2$. The particles are a group of real scalars independently drawn from a certain distribution. Consider that $N \in \mathbb{N}^+$ particles $\mathcal{Z} = \{z^{(1)}, z^{(2)}, \cdots, z^{(N)}\}$, $z^{(i)} \in \mathbb{R}$, $i = 1,2,\cdots,N$, are independently drawn from a distribution depicted by a PDF $p(\cdot)$. Then, the unbiased estimation of $p(\cdot)$ can be obtained using a Gaussian kernel method as
\begin{equation}\label{eq:g_kernel}
\hat{p}(z, \mathcal{Z}) = \frac{1}{\sqrt{2\pi}\hat{\sigma} N} \sum_{j=1}^{N} \mathrm{exp} \left(-\frac{\left(z-z^{(j)}\right)^2}{2\hat{\sigma}^2}\right),
\end{equation}
where $\hat{\sigma} \in \mathbb{R}^+$ is a variance parameter. Here, we use the symbol $\hat{p}(\cdot)$ to represent the PDFs approximated using particles. Based on the approximated PDFs $\hat{p}_{\hat{e}_i}(\cdot)$, the predictive coefficients $\overline{P}_{i+1}$ are calculated recursively, following the flow $\hat{p}_{\hat{e}_{i-1}}(\cdot) \rightarrow \hat{p}_{\hat{e}^{\eta}_{i-1}}(\cdot) \rightarrow \hat{p}_{\hat{e}_i}(\cdot) \rightarrow \overline{P}_{i+1}$. 
 
The approximation in each iteration is described as follows. In line 4, the particles exceeding the threshold $\eta$ are removed, which simulates the truncation operation to the PDF $p_{\hat{e}_{i-1}}(\cdot)$. Then, in line 5, the PDF $p_{\hat{e}^{\eta}_{i-1}}(\cdot)$ of the truncated stochastic variable $\hat{e}^{\eta}_{i-1}$ is approximated with the remaining particles. In line 6, $N$ particles are resampled from the approximated PDF $\hat{p}_{\hat{e}^{\eta}_{i-1}}(\cdot)$. The particles then perform the statistical propagation according to the error dynamics \eqref{eq:close}, as shown in lines 7-10. In line 11, the particle approximation method~\eqref{eq:g_kernel} is used again to approximate the PDF $p_{\hat{e}_i}(\cdot)$. Finally, the predictive coefficient $\overline{P}_{i+1}$ are calculated in line 14. 

\begin{algorithm}[htbp]
\caption{Approximation of the predictive coefficients using particles}
\label{ag:rme}
\SetKwData{Left}{left}\SetKwData{This}{this}\SetKwData{Up}{up}
\SetKwFunction{Union}{Union}\SetKwFunction{FindCompress}{FindCompress}
\SetKwInOut{Input}{Inputs}\SetKwInOut{Output}{Outputs}
  
\Input{noise variance $\sigma$ and particle number $N$}
\Output{$\overline{P}_i,~\forall \, i =1,2,\cdots,T$, $T>2$}

Calculate $\overline{P}_1, \overline{P}_2$ using \eqref{eq:con2} with $p_{\hat{e}_0}(\cdot)$, $p_{\hat{e}_1}(\cdot)$ in \eqref{eq:pdf_0_1}\;
Sample particles $z_1^{(j)} \sim \mathcal{N}(0,\sigma)$, $j=1,2,\cdots,N$\;

\For{$i \leftarrow 2$ \KwTo $T-1$}{
    Remove all particles $\left|z_{i-1}^{(j)}\right| \geq \eta$\;
    Approximate $\hat{p}_{\hat{e}^{\eta}_{i-1}}(\cdot)$ with $z_{i-1}^{(j)}$ using \eqref{eq:g_kernel}\;
    
    Re-sample particles $z_{i-1}^{(j)} \sim \hat{p}_{\hat{e}^{\eta}_{i-1}}(\cdot)$; $j=1,2,\cdots,N$\;
    \For{$j\leftarrow 1$ \KwTo $N$}{
        Draw $\epsilon_{i-1}^{(j)} \sim \mathcal{N}(0,\sigma)$\;
        $z_{i}^{(j)} =A z_{i-1}^{(j)} +\epsilon_{i-1}^{(j)}$\;
        }
        Approximate $p_{\hat{e}_{i}}(\cdot)$ with $z_{i}^{(j)}$ using \eqref{eq:g_kernel}\;
        Calculate $\overline{P}_{i+1}$ with \eqref{eq:con2} using PDF $p_{\hat{e}_{i}}(\cdot)$\;
        }
    
\end{algorithm}

Note that Algorithm \ref{ag:rme} may lead to approximation errors in the predictive coefficients. The main source of the errors is the deviation between the PDFs $p_{\hat{e}_i}(\cdot)$ and their estimations $\hat{p}_{\hat{e}_i}(\cdot)$. In fact, the unbiasedness of the approximation only holds in the statistical sense. To reduce the approximation errors, $N$ should be selected sufficiently large, and $\hat{\sigma}$ should be small.

\begin{remark}
In this paper, our theoretical claims and numerical methods target a specific class of NET-SCS, where the network communication is triggered by an asynchronous event associated with state estimation errors. In fact, the state estimation error $e_k$, $k \in \mathbb{N}$, can be recognized as a variable that depends on the internal states of the joint dynamic model of the system plant and the state estimator, namely the plant state $x_k$ and the estimator state $\hat{x}_k$. Thus, our results can also be extended to a generic NET-SCS of which the triggering event may be assigned to an arbitrary state-dependent variable. In this case, the recursive model of the ACR is still effective. What changes is that the predictive coefficients are calculated using the PDF of this state-dependent variable. The challenge of such an extension depends on the complexity of this PDF. 
\end{remark}

\rv{Now, we briefly discuss the possible extension of the results of this paper to a multi-dimensional NET-SCS. The main challenge preventing this extension is the derivation of a multi-dimensional ACR recursive model for a multi-dimensional NET-SCS based on a proper multi-variable event-triggered scheduler, which remains an open question. Once a multi-dimensional ACR recursive model is obtained, it is possible to extend Algorithm~\ref{ag:rme} to approximate the multi-dimensional ACR numerically by substituting the sampling distribution $\mathcal{N}(0, \sigma)$ with a multi-variable PDF.}

\section{Comparison with the Conventional Method}\label{sec:compa}

Based on Sec.~\ref{sec:main_result_delay_free} and Sec.~\ref{sec:coeffi}, we are able to calculate the stationary and the transient ACR for a NET-SCS using a finite number of predictive coefficients. The analytical and numerical methods to compute these coefficients are also provided. In this section, we make a comparison between our approaches and the conventional method~\cite{ebner2016communication} that intentionally ignores the side information for simplification. Both theoretical analysis and a numerical study are conducted to validate the accuracy of our approaches and qualitatively verify the accuracy gap between the conventional method and the ground truth.

\subsection{Deviation Analysis of The Conventional Method} \label{sec:comp_gaussian}

As mentioned above, the computation of ACR without incorporating the side information leads to an oversimplified distribution model for the state estimation error and eventually returns approximated results. 
Assume that the open-loop state estimation error is subject to the dynamic model~\eqref{eq:open}. Then, the error has a fully Gaussian PDF, and, similar to \eqref{eq:expcon}, the open-loop ACR can be recursively computed as
\begin{equation}\label{eq:expcon_2}
\textstyle \breve{\E}\left(\delta_k \right) = 
1 - \sum^{\min(k,T)}_{n=1} \prod_{j=1}^n \breve{\overline{P}}_{j} \breve{\E}\left(\delta_{k-n} \right),~k\in \mathbb{N}^+,
\end{equation} 
with $\breve{E}(\delta_0) = 1$, where $\breve{\overline{P}}_i$, $i=1,2,\cdots,T,$ are the coefficients obtained by
\begin{equation}\label{eq:conv_2}
\textstyle \breve{\overline{P}}_i =\int^{\eta}_{-\eta} p_{e_{i-1}}(z) \mathrm{d} z,~~i=1,2,\cdots,T,
\end{equation}
where $p_{e_i}(\cdot)$ is the PDF of the open-loop state estimation error $e_i$ subject to the dynamic model~\eqref{eq:open} with $\iota =0$. Hence, according to \eqref{eq:open}, for all $i>1$, we have
\begin{equation}\label{eq:pizi1}
\begin{split}
p_{e_i}(z) = &\textstyle  \, \int_{-\infty}^{\infty}p_{e_{i-1}}(\xi)\;p_w(z - A \xi) \mathrm{d} \xi \\
= & \, \int_{-\infty}^{\infty}\frac{p_{e_{i-1}}(\xi)}{\sqrt{2\pi}\sigma} \;\mathrm{exp} \left(-\frac{(z-A\xi)^2}{2\sigma^2}\right) \mathrm{d} \xi,
\end{split}
\end{equation}
with the initial conditions 
\begin{equation}\label{eq:pizi0}
p_{e_0}(z) = \delta(z),~~~p_{e_1}(z) = 
\frac{1}{\sqrt{2\pi}\sigma} \;\mathrm{exp} \left( - \frac{z^2}{2\sigma^2} \right).
\end{equation}
Comparing (\ref{eq:detui1}) and (\ref{eq:pizi1}), one notices that $\hat{e}_1$ and $e_1$ have the same distribution $\mathcal{N}(0,\sigma)$, while for each $i > 1$, $p_{\hat{e}_i}(\cdot)$ has an additional multiplier $G_{\hat{e}_{i-1}}(\cdot)$, compared to $p_{e_i}(\cdot)$. Also, the integration intervals are also different.

Now, we compare the mean values and the variances of the two stochastic variables $\hat{e}_i$ and $e_i$ for $i = 1,2,\cdots,T$. From~\eqref{eq:open1}, we know that the state estimation error $e_i$ is a linear combination of the Gaussian-distributed stochastic variables $w_0$, $\cdots$, $w_{i-1}$. Hence, $e_i$ is also Gaussian-distributed and has the following property. 

\begin{property}\label{pp:oldi}
Given that $w_{0}$, $\cdots$, $w_{T-1}$ are i.i.d. stochastic variables (Assumption \ref{as:as1}), the following statements hold for all $e_i$, $i=1,2,\cdots,T$.
\begin{enumerate}
\item \label{it:pek} For any $z \in \mathbb{R}$, $p_{e_i}(z) = p_{e_i}(-z)$.\vspace{1.5mm}
\item \label{it:Eek}$\E\!\left( e_i \right) = \sum_{n=\iota}^{i-1} A^{i-n-1} \E(w_n) = 0$.\vspace{1.5mm}
\item \label{it:Vek} $\mathrm{Var}\!\left( e_i \right) = \sum_{n=\iota}^{i-1} A^{i-n-1} \mathrm{Var}\!\left(w_n\right) = \sum_{n=\iota}^{i-1} A^{i-n-1} \sigma^2$.
\end{enumerate}
\end{property}

Property~\ref{pp:oldi} is easy to verify using the linear properties of Gaussian stochastic variables. Nevertheless, the stochastic properties of the state estimation error $\hat{e}_i$ are not straightforward due to the recursive truncation operations. Before proceeding with the study on the stochastic properties of $\hat{e}_i$, it is necessary to propose the following proposition for truncated stochastic variables.

\begin{proposition}\label{pp:trunc}
Let $\zeta \in \mathbb{R}$ be an arbitrary stochastic variable of which the PDF $p_{\zeta}(z)$ has infinite support. $\E(\zeta)$ and $\mathrm{Var}(\zeta)$ are respectively its mean value and variance. Also, let $\zeta^{\eta} \in \mathbb{R}$ be a truncated stochastic variable by trimming the support of $\zeta$ to be within the symmetrically bilateral interval $[\,-\eta,\eta\,]$, $\eta>0$. If $\E(\zeta)=0$, and $p_{\zeta}(z) = p_{\zeta}(-z)$ holds for all $z \in \mathbb{R}$, then the following conditions are valid.
\begin{enumerate}
\item $\E(\zeta^{\eta}) = 0$, and $p_{\zeta^{\eta}}(z) = p_{\zeta^{\eta}}(-z)$, $\forall \, z \in \mathbb{R}$.
\item $\mathrm{Var}(\zeta^{\eta}) < \mathrm{Var}(\zeta)$.
\end{enumerate}
\end{proposition}

\begin{proof}
See Appx.-D for the details of this proof.
\end{proof}

Proposition~\ref{pp:trunc} indicates that a truncated stochastic variable has the same expected value but a smaller variance than its original counterpart if the latter has an even PDF around zero and the truncation interval is symmetric. We now present the following theorem that characterizes the relation between the mean values and variances of the state estimation errors.

\begin{theorem}\label{pp:e0}
Given state estimation errors $\hat{e}_i$ and $e_i$ depicted by the dynamic models~(\ref{eq:close}) and~(\ref{eq:open}), respectively, $i = 1,2,\cdots,T$, the following conditions hold.
\begin{enumerate}
\item $p_{\hat{e}_i}(z) = p_{\hat{e}_i}(-z)$, for all $z \in \mathbb{R}$.
\item $\E\!\left(\hat{e}_i \right) = \E(e_i) = 0$.
\item $\mathrm{Var}(\hat{e}_1) = \mathrm{Var}(e_1)$, $\mathrm{Var}(\hat{e}_i) < \mathrm{Var}(e_i)$, for all $i \geq 2$.
\end{enumerate}
\end{theorem}
\begin{proof}
See Appx.-E for the details of this proof.
\end{proof}

Theorem~\ref{pp:e0} indicates the qualitative difference between the PDFs, the mean values, and the variances of the estimation error $\hat{e}_i$ and the open-loop error $e_i$ for $i=1,2,\cdots,T$. Both of them have even PDFs and zero mean values. Nevertheless, the estimation error $\hat{e}_i$ has a smaller variance than the open-loop one $e_i$, for $i>1$. This indicates that the recursive truncation operations in~\eqref{eq:close} result in a \textit{shrink} in the PDF $p_{\hat{e}_i}(z)$ along the $z$-axis compared to the infinite support Gaussian PDF $p_{e_i}(z)$. Therefore, for any $i>1$, Theorem~\ref{pp:e0} results in
\begin{equation*}
\textstyle \int^{\eta}_{-\eta} p_{\hat{e}_i} (z)\mathrm{d}z > \int^{\eta}_{-\eta} p_{e_i} (z)\mathrm{d}z.
\end{equation*}
This can be explained intuitively that the shape of $p_{\hat{e}_i}(\cdot)$ is more \textit{narrow} than $p_{e_i}(\cdot)$. Based on this, we can infer that the conventional method using $p_{e_i}(\cdot)$ instead of $p_{\hat{e}_i}(\cdot)$ leads to smaller results for the coefficients, i.e., $\breve{\overline{P}}_{i+1} < \overline{P}_{i+1}$ for $i=3,\cdots,T$, according to \eqref{eq:con2}, and then larger values of the transient ACR, i.e., $\breve{\E}(\delta_k) < \E(\delta_k)$ for $k=3,\cdots,T$. Extending this claim to $k \rightarrow \infty$, we also have a similar conclusion for the stationary ACR, i.e., $\breve{\E}(\delta_{\infty}) < \E(\delta_{\infty})$.

The analysis in this section not only proves the accuracy gap of the conventional method in theory but also qualitatively points out that it always leads to larger computation results. 

\subsection{Accuracy Comparison: A Numerical Example}\label{sec:eg}
We use a numerical example to verify the accuracy of our proposed methods in Sec.~\ref{sec:crc} and Sec.~\ref{sec:nm}, respectively. We also validate the accuracy gap of the conventional method that ignores the close-loop effect. 
Consider a NET-SCS as in~\eqref{eq:plant} with parameters $A = 1.25$, $B = 1$, an initial state $x_0 = -2$, a stochastic process $w_k \sim \mathcal{N}(0,1)$, $k \in \mathbb{N}^+$, and a state-feedback controller $u_{k} = -\hat{x}_{k}$, where $\hat{x}_k$ is estimated using~\eqref{eq:est}. The threshold and the maximum triggering interval of the event-triggered scheduler~\eqref{eq:schedul_variables_policy} are $\eta=1$, $T=5$.  As addressed in Sec.~\ref{sec:comp_gaussian}, the major difference between our work and the existing works is that the latter ignores the side information of a NET-SCS and uses the open-loop estimation error $e_k$ to compute ACR, instead of the estimation error $\hat{e}_k$. 
We use five manners to compute the transient ACR to provide a fair and clear comparison study.

\subsubsection{The Proposed Analytical Method (PAM)}
The recursive expressions \eqref{eq:pdf_0_1} and \eqref{eq:detui1} are used to obtain the PDFs $p_{\hat{e}_i}(\cdot)$ of the state estimation error $\hat{e}_i$ for $i=0,1,\cdots,4$. Then, the coefficients $\overline{P}_{i+1}$ are calculated using \eqref{eq:con2}. Finally, \eqref{eq:expcon} is recursively used to compute the transient ACR $\E(\delta_k)$ for $k=1,2,\cdots,5$.
An example of this calculation is provided in Appx.-F, where we only give the results for $k=1,2,3$, due to the complexity of analytical computation for larger $k$ values.

\subsubsection{The Proposed Numerical Method (PNM)}

Algorithm~\ref{ag:rme} is used to approximate the PDFs $p_{\hat{e}_i}(\cdot)$ of the state estimation errors $\hat{e}_i$ for $i=0,1,\cdots,4$, with parameters $\hat{\sigma}=0.1$ and $N=10^4$. Then, the coefficients $\overline{P}_{i+1}$ are calculated using \eqref{eq:con2}. Finally, \eqref{eq:expcon} is recursively used to compute the transient ACR $\E(\delta_k)$ for $k=1,2,\cdots,5$.

\subsubsection{The Conventional Analytical Method (CAM)}

\cite{ebner2016communication}~The recursive expressions \eqref{eq:pizi1} and \eqref{eq:pizi0} are used to obtain the PDFs $p_{e_i}(\cdot)$ of the open-loop state estimation errors $e_i$ for $i=0,1,\cdots,4$. Then, the open-loop predictive coefficients $\breve{\overline{P}}_{i+1}$ are calculated using \eqref{eq:conv_2}. Finally, \eqref{eq:expcon_2} is recursively used to compute the transient ACR $\breve{\E}(\delta_k)$ for $k=1,2,\cdots,5$.

\subsubsection{The Conventional Numerical Method (CNM)}

This approach is merely used to provide a numerical counterpart of the CAM approach for the completeness of our work.
We first use Algorithm~\ref{ag:rme}, with the same parameters $\hat{\sigma}=0.1$ and $N=10^4$ as PNM but with the lines 4-6 removed, to calculate the open-loop predictive coefficients $\breve{\overline{P}}_{i+1}$. Then, \eqref{eq:expcon_2} is recursively used to compute the transient ACR $\breve{\E}(\delta_k)$ for $k=1,2,\cdots,5$.

\subsubsection{Ground Truth (GT)}

We conduct a Monte-Carlo experiment of the NET-SCS with the same initial state repeated for $10^4$ trials to approximate the true value of ACR,
\begin{equation*}
\E_{GT}(\delta_k) = \#(\delta_k=1)/10^{4},
\end{equation*}
where $\#(\delta_k=1)$ is the total number of trials of which $\delta_k = 1$.

\begin{figure*}[htbp]
\centering
\subfloat[$k=2$]
{\includegraphics[width=0.23\textwidth]{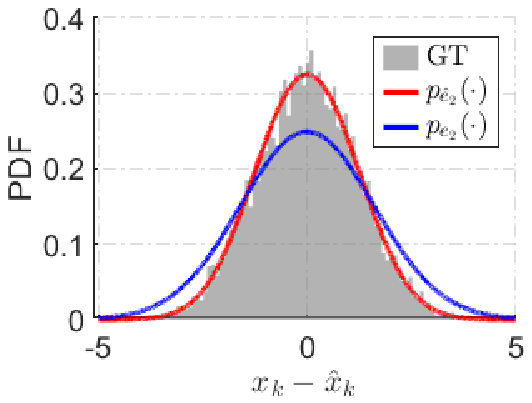}}
\subfloat[$k=3$]{\includegraphics[width=0.23\textwidth]{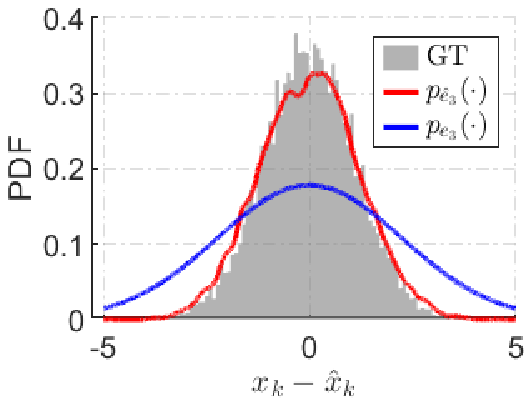}}
\qquad
\subfloat[$k=4$]{\includegraphics[width=0.23\textwidth]{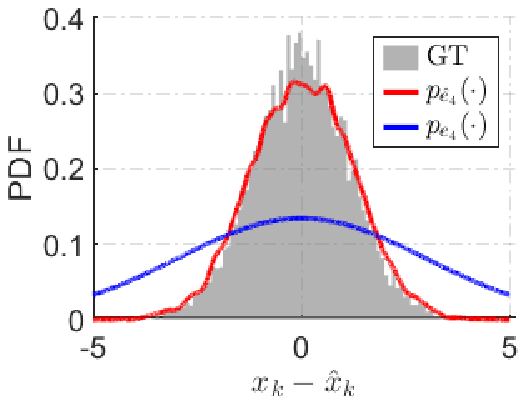}}
\subfloat[$k=5$]{\includegraphics[width=0.23\textwidth]{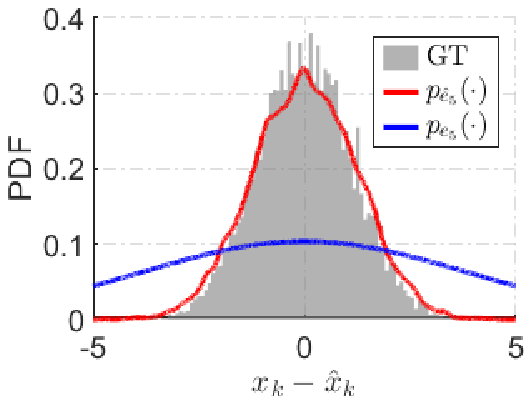}}
\caption{The approximated PDFs of the state estimation errors for $k \!=\! 2,3,4,5$. The Red line is $\hat{p}_{\hat{e}_k}(\cdot)$ calculated using our proposed methods ($k=2$ using PAM and $k=3,4,5$ using PNM) and the blue line is $\hat{p}_{e_k}(\cdot)$ obtained from the conventional approach (CAM). The gray area represents the GT PDF using Monte Carlo sampling.}
\label{fig:fig2}
\end{figure*}

The PDFs calculated using $p_{\hat{e}_k}(\cdot)$ and $p_{e_k}(\cdot)$, for $k=2,3,4,5$, are illustrated in Fig.~\ref{fig:fig2}, in red and blue, respectively. The GT PDF of the state estimation errors, drawn as the gray area, obtained by conducting a Monte Carlo experiment, is also presented for comparison. We observe that our proposed method accurately follows the GT. On the contrary, the conventional method obviously deviates from the GT results. The deviation becomes larger as $k$ increases. This qualitatively verifies our theoretical claims in Sec.~\ref{sec:comp_gaussian}.

Our theoretical claims can also be justified by quantitative results from the numerical study, as presented in Table~\ref{tab:acr}. Slight deviations are seen between PAM and PNM or between CAM and CNM. Note that these deviations reflect the inevitable approximation errors between the analytical methods and their numerical counterparts due to the approximation bias of the Gaussian kernel method. Incorporating these errors, we can see that the results of PAM and PNM are very close to the ground truth (with absolute errors smaller than $0.005$), validating the accuracy of the proposed methods. On the contrary, the results of CAM and CNM present large calculation errors. Moreover, they are all larger than the GT results, verifying our theoretical arguments in Sec.~\ref{sec:comp_gaussian} that the conventional method overapproximates the ACR values.

\linespread{1.2}
\begin{table}[htbp]
\caption{The GT and the computed ACR values for NET-SCS}
\label{tab:acr}
\begin{center}
\begin{tabular}{c|c|c|c|c|c}
\hline
$k$ & GT & PAM & PNM & CAM & CNM  \\\hline
1 & $0$ & $0$ & $0$ & $0$ & $0$     \\
2 & $0.3175$ & $0.3173$ & $0.3129$ & $0.3173$ & $0.3161$      \\
3 & $0.2826$ & $0.2818$ & $0.2877$ & $0.3633$ &  $0.3668$     \\
4 & $0.2650$ & --- & $0.2609$ & $0.3098$  & $0.3082$\\
5 & $0.2801$ & --- &$0.2797$ & $0.3117$ &$0.3126$\\
\hline
\end{tabular}
\end{center}
\end{table}
\linespread{1}

More details can be found by looking into the stochastic properties of the state estimation errors. Table~\ref{tab:ev} shows the mean values $\E(\cdot)$ and the variances $\mathrm{Var}(\cdot)$ of the estimation error $\hat{e}_k$ and open-loop error $e_k$ for $k=1,\cdots,5$. It can be seen that their mean values are very close to zero, despite small errors due to the numerical approximation. Also, we witness $\mathrm{Var}(e_1)=\mathrm{Var}(\hat{e}_1)=0$ and $\mathrm{Var}(e_k)>\mathrm{Var}(\hat{e}_k)$, for all $k=2,\cdots,5$. This coincides with our theoretical statements in Theorem~\ref{pp:e0} that the open-loop errors have the same mean values as the estimation errors but larger variances. 

\linespread{1.2}
\begin{table}[htbp]
\caption{The mean values and variances of state estimation errors}
\label{tab:ev}
\begin{center}
\begin{tabular}{c|c|c|c|c}
\hline
$k$ & $\E(\hat{e}_k)$ & $\E(e_k)$ & $\mathrm{Var}(\hat{e}_k)$ & $\mathrm{Var}(e_k)$  \\\hline
1 & $0$ & $0$  & $0$ & $0$     \\
2 & $0.0000$ & $-0.0000$  & $1.4549$ & $2.5625$     \\
3 & $-0.0037$ & $-0.0000$ & $1.4497$ & $5.0031$  \\
4 & $0.0220$ & $-0.0000$  & $1.5108$ & $8.7302$      \\
5 & $0.0149$ & $0.0000$ & $1.5288$ & $13.601$      \\
\hline
\end{tabular}
\end{center}
\end{table}
\linespread{1}

\section{Experimental Study} \label{sec:simulation}

In this section, we conduct an experimental study of a leader-follower autonomous driving scenario, as illustrated in Fig.~\ref{fig:case}, to validate our theoretical results interpreted so far.
The leader-follower scenario is a simplified case of the widely-used platooning model in autonomous driving~\cite{liu2020cloud}. A leading vehicle drives along a predefined trajectory and each follower should keep a constant distance from the leader. The positions and velocities of the vehicles are measured using remote sensors, connected to the vehicles with a common communication network which is activated by an event-based scheme as Eq.~\eqref{eq:schedul_variables_policy}.

\begin{figure}[htbp]
\noindent
\hspace*{\fill} 
\begin{tikzpicture}[scale=1,font=\small]

\node[anchor=south] (tower) at (0.4cm,-0.2cm) {\includegraphics[height=2cm]{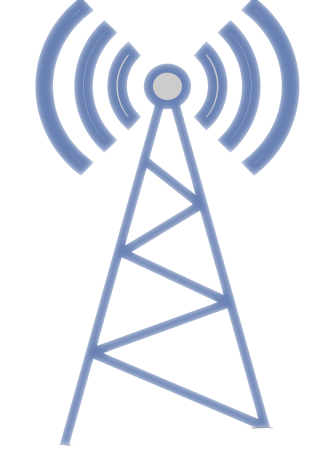}};
\node[anchor=west, text width=3cm, align=center] () at ([xshift=-0.7cm, yshift=0.5cm] tower.east) {\small Communication network};

\node[anchor=south] (pole) at (-2.5cm,-0.2cm) {\includegraphics[height=1.5cm]{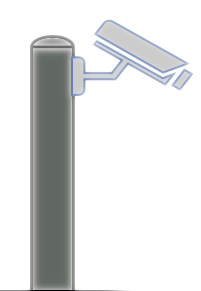}};
\node[anchor=center, text width=3cm, align=center] () at (pole.north) {\small Remote sensor};

\node[] (tractors) at (0cm,0cm) {\includegraphics[height=1cm]{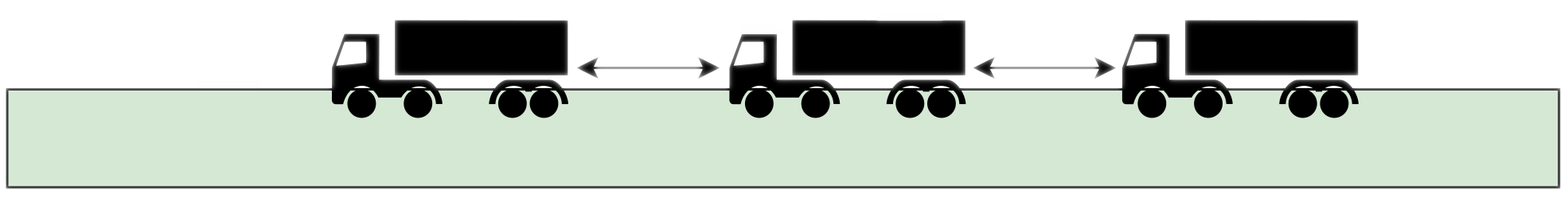}};
\node[] () at ([xshift=-1.5cm] tractors.north) {\small Leader};
\node[] () at ([xshift=1.2cm, yshift=-0.87cm] tractors.north) {\small Followers};
\node[anchor=north, align=center] () at ([xshift=-0.7cm] tractors.north) {\small $d$};
\node[anchor=north, align=center] () at ([xshift=1.3cm] tractors.north) {\small $d$};

\node[anchor=south] (wifi1) at (-0.15cm,0.3cm) {\includegraphics[height=0.3cm]{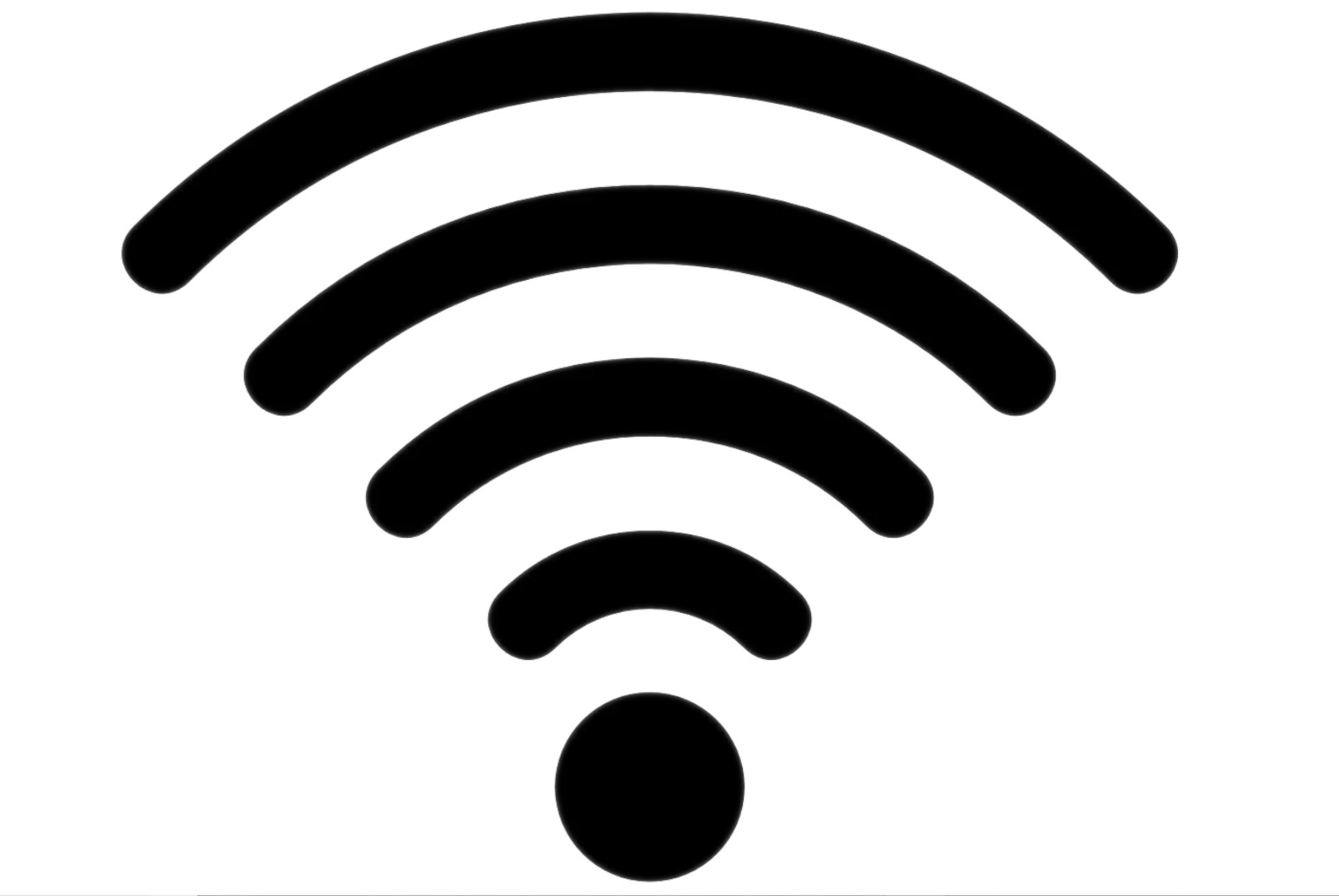}};
\node[anchor=south] (wifi2) at (1.85cm,0.3cm) {\includegraphics[height=0.3cm]{wifii.png}};

\end{tikzpicture}
\hspace{\fill} 
\caption{\rv{A leader-follower autonomous driving scenario.}}
\label{fig:case}
\end{figure}

The position of the leader vehicle follows a predefined trajectory $p^L(t) = -\cos(t) + 1.2t$. The follower is required to maintain a distance $d = 3\,$m with the leader. The kinematic model of the follower vehicle is
\begin{equation}
\begin{split}
\dot{p}(t) \, &= v(t), \quad\;
\dot{v}(t) \, = u(t),
\label{eq:double_integrator}
\end{split}
\end{equation}
where $p(t), v(t), u(t) \in \mathbb{R}$ are respectively the position, the velocity, and the acceleration of the follower. In this experiment, the parameters are selected as $\gamma =1$, $Q = 1$, and $K = 1$. 
The objective of the problem is to design a control law $u(t)$, such that $p(t) \rightarrow p^L(t)+d$ and $v(t)\rightarrow\dot{p}^L(t)$ as time $t$ increases. We define the following feedback control law,
\begin{equation}\label{eq:ctrl}
\begin{split}
u(t) &= - \gamma Q^{-1} v(t) - Q^{-1}Kp(t) + \gamma Q^{-1} \dot{p}^L(t)\\
& ~~~ + Q^{-1}Kp^L(t) + \ddot{p}^L(t) +Q^{-1}Kd,  
\end{split}
\end{equation}
where $K, Q, \gamma \in \mathbb{R}^+$ are positive parameters. It can be verified using a Lyapunov method that $p(t)- p^L(t) = d$ and $v(t) - \dot{p}^L(t) =0$ render a globally asymptotic equilibrium of the closed-loop system, which indicates the achievement of the desired control performance. The proof is omitted in this paper.

In our experiment, we consider the discrete-time version of the follower vehicle in Eq.~\eqref{eq:double_integrator},
\begin{equation*}
\begin{split}
p_{k+1} &= p_k + \Delta{t}v_k,\\
v_{k+1} &= v_k + \Delta{t}u_k + w_k,\\
\end{split}
\end{equation*}
where $\Delta{t}$ is the sampling period, and $w_k$ is an i.i.d. noise process. Accordingly, we discretize the reference trajectory $p^{L}(t)$ to $p^L_k$ using discrete sampling $t = k\cdot \Delta t$. 
In correspondence with the NET-SCS model in Fig.~\ref{fig:block_diagram}, the leader's trajectory $p^L_k$ is the reference signal of the overall system. Each follower is a plant with the state $x_k \!=\! v_k$. The limited communication bandwidth motivates the application of the event-triggered scheduler in \eqref{eq:schedul_variables_policy}, for which we set $\eta \!=\! 1$, and $T\!=\!20$ in this experiment. The state estimator \eqref{eq:est} is used to obtain $\hat{x}_k \!=\! \hat{v}_k$, with $A \!=\!1$, and $B \!=\!\Delta t$. The discrete-time controller based on the estimated state is $u_k = u(k \cdot \Delta t)$, for which $\hat{p}_k$ is obtained using the recursive model $\hat{p}_{k+1} = \hat{p}_k + \Delta{t}\hat{v}_k$.
The simulation runs for $10^4$ trials with the same initial conditions $p_0 = 0$ and $v_0 = 0$. Each trial lasts for $t = 40\,s$ with a sampling time $\Delta t = 0.1s$. The overall control performance is shown in Fig.~\ref{fig:tracking}. It is observed that the average following distance $E(p(t)) - p^L(t)$ slightly fluctuates around $d=3\,$m, which indicates satisfactory distance keeping. Also, the average velocity 
$E(v(t))$ is very close to the reference velocity $\dot{p}^L(t)$. This shows that the configuration of the state estimator~\eqref{eq:est} and the event-triggered scheduler \eqref{eq:schedul_variables_policy} successfully achieves the control objectives.

\begin{figure}[htbp]
\centering
\subfloat[tracking distance]
{\includegraphics[width=4.5cm]{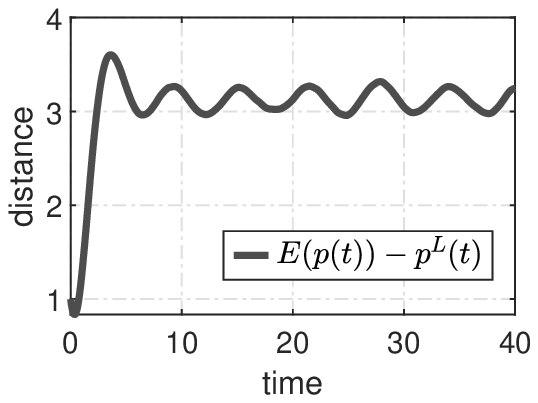}}
\subfloat[velocity tracking]{\includegraphics[width=4.5cm]{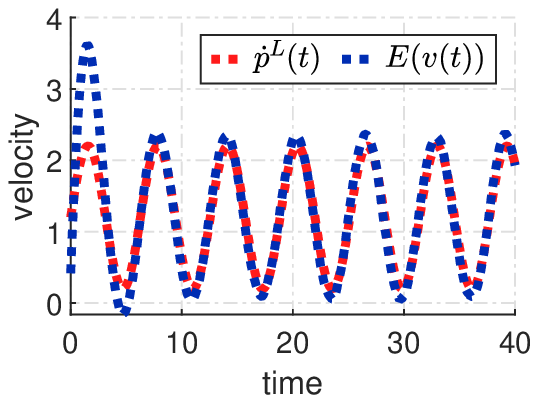}}
\caption{Average performance of the platoon controller~\eqref{eq:ctrl}. Plot~(a) depicts the tracking distance $\E(p(t)) -p^L(t)$. Plot~(b) shows the leading velocity $\dot{p}^L(t)$ (in red) and the mean of the actual velocity $\E(v(t))$ (in blue).}
\label{fig:tracking}
\end{figure}

The computed values of the ACR $\E(\delta_k)$, using Algorithm~\ref{ag:rme} (PNM), with various triggering thresholds $\eta$, are shown in Fig.~\ref{fig:comm_rate_evolution} (in red). To verify the validity of our proposed method, we also show the GT-ACR obtained from Monte-Carlo simulation (in black), and the ACR computed according to the conventional method (CAM), i.e., $\breve{E}(\delta_k)$ (in blue). The information delivered by Fig.~\ref{fig:comm_rate_evolution} can be summarized as follows.

\subsubsection{The general existence of the stationary ACR} 
It is noticed that all ACR values, $\E(\delta_k)$, $\breve{\E}(\delta_k)$, and the GT-ACR, ultimately converge to their respective stationary points for all triggering threshold values $\eta=1,2,3,4$. This validates our result on the existence of the stationary ACR in Sec.~\ref{sec:conv_ana}. 

\subsubsection{The accuracy of the proposed method} 
It is observed that the computed ACR $\E(\delta_k)$ closely follows the GT-ACR at all times, indicating the accuracy of our proposed method. On the contrary, $\breve{\E}(\delta_k)$ shows deviations from the GT-ACR suggesting the inaccuracy of computing the ACR by this method. Also, $\breve{\E}(\delta_k)$ is in general larger than $\E(\delta_k)$ in the steady state, which validates our theoretical statement in Sec.~\ref{sec:comp_gaussian} that this method overestimates the stationary ACR. 

\subsubsection{The influence of the triggering threshold} 
The stationary ACR values tend to be smaller as the triggering threshold $\eta$ increases. The intuition behind this observation is that a higher threshold means higher estimation errors are tolerable, hence fewer events will be triggered to reset the estimation error, which consequently leads to lower ACR. A similar observation is depicted in Fig.~\ref{fig:threshold_vs_rate}, where the change of stationary ACRs $\E(\delta_{\infty})$ and $\breve{\E}(\delta_k)$, and their ratios are plotted versus the changes of the threshold $\eta$. It van be seen that $\E(\delta_{\infty}) < \breve{\E}(\delta_k)$, for all values of $\eta$. However, the scale of the deviation between the two approaches is not monotone with respect to $\eta$, i.e., larger triggering thresholds do not necessarily lead to larger deviations. The largest deviation occurs around $\eta=3$, with more than $25\%$, which is noticeable.

\begin{figure}[htbp]
    \centering
    \begin{subfigure}[b]{0.23\textwidth}
        \includegraphics[width=\textwidth]{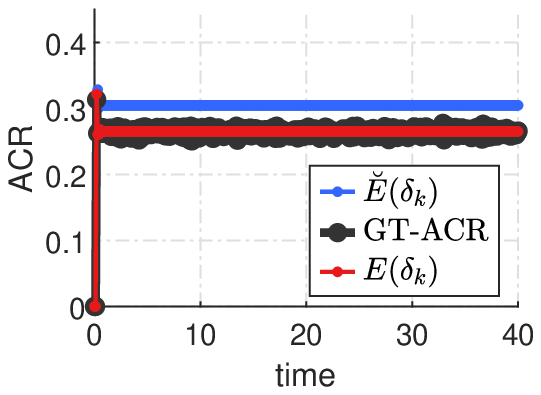}
        \caption{$\eta=1$}
        \label{labeleta.1}
    \end{subfigure}
    \hfill
    \begin{subfigure}[b]{0.23\textwidth}
	\includegraphics[width=\textwidth]{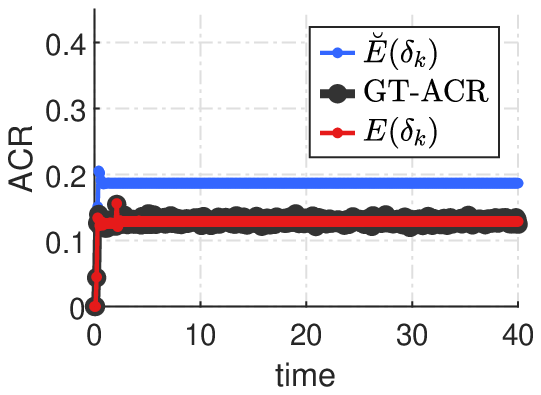}
        \caption{$\eta=2$}
        \label{labeleta.2}
    \end{subfigure}
    \hfill
    \begin{subfigure}[b]{0.23\textwidth}
        \includegraphics[width=\textwidth]{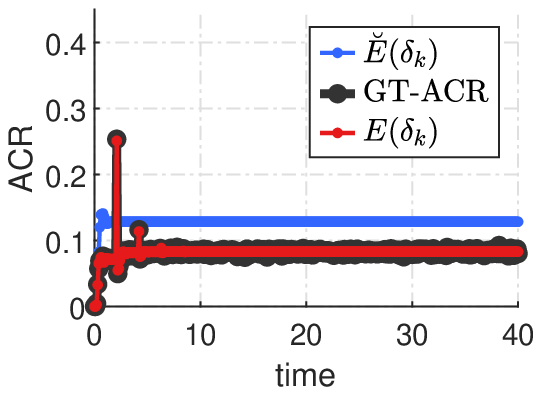}
        \caption{$\eta=3$}
        \label{labeleta.3}
    \end{subfigure}
    \hfill
    \begin{subfigure}[b]{0.23\textwidth}
        \includegraphics[width=\textwidth]{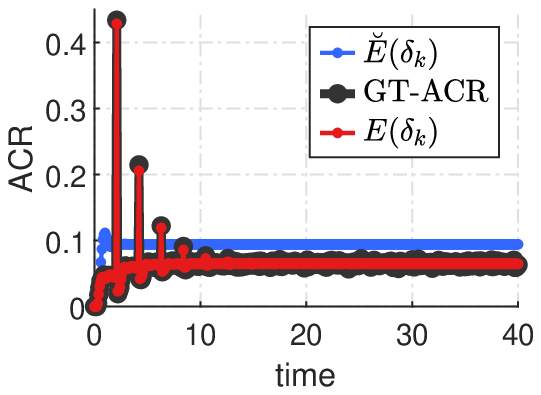}
        \caption{$\eta=4$}
        \label{labeleta.4}
    \end{subfigure}
    \caption{ACRs computed using our proposed method $\E(\delta_k)$ (in red), the existing method $\breve{\E}(\delta_k)$ (in blue), and GT-ACR (in black), for various triggering thresholds $\eta=1,2,3,4$.}
    \label{fig:comm_rate_evolution}
\end{figure}

\begin{figure}[htbp]
\centering
\begin{subfigure}[b]{0.23\textwidth}
    \includegraphics[width=\textwidth]{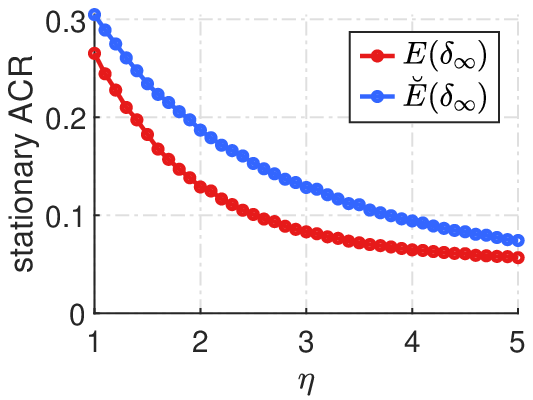}
    \caption{Stationary ACRs}
    \label{label1.1}
\end{subfigure}
\hfill
\begin{subfigure}[b]{0.23\textwidth}
\includegraphics[width=\textwidth]{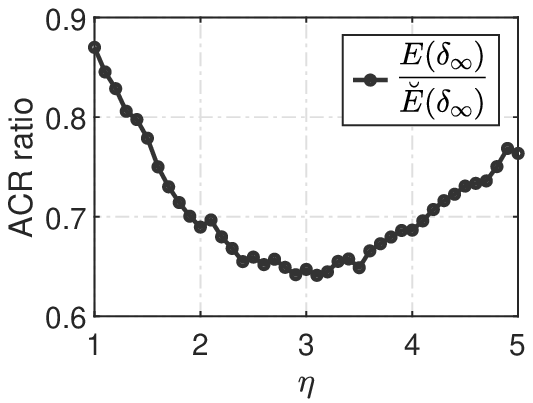}
    \caption{ACR ratio}
    \label{label1.2}
\end{subfigure}
\caption{Comparison between the stationary ACR vs. triggering threshold. Plot (a), our proposed method $\E(\delta_{\infty})$ (in red), and the existing method $\breve{\E}(\delta_{\infty})$ (in blue). Plot (b) shows the ratio of the two stationary ACR vs. triggering thresholds.}
\label{fig:threshold_vs_rate}
\end{figure}

\section{Conclusion} \label{sec:conclusions}

Motivated by the conservativeness of conventional methods, we provide comprehensive analytical formulations to compute the accurate ACR of a NET-SCS. We derive a novel recursive ACR model by precisely tracking the statistical propagation of truncated probability distributions. Based on this model, we propose analytical and numerical approaches to accurately calculate ACR values and showcase the noticeable ACR over-estimation when the triggering-induced truncations are ignored in computing the ACR. 
The efficacy of our proposed method and theory is validated with a numerical example and an experimental study on a platooning scenario, showing that our ACR computation model precisely follows the ground truth case. Future work will focus on using the accurate ACR to facilitate the design of efficient networked systems.

\section*{Acknowledgement}
The authors would like to thank Prof. Biqiang Mu and Prof. Hongsheng Qi from the Chinese Academy of Sciences for their valuable discussions on truncation analysis. The authors would also like to thank Prof. Yirui Cong for helpful discussions on the communication rate analysis.

\section*{Appendix}

\subsection{Existence of Steady State of A Discrete-Time LTI System}

In this paper, we construct a recursive model to depict the timed evolution of the ACR. The recursive ACR model is indeed a discrete-time LTI (dt-LTI) system and the stationary ACR is equivalent to its steady state. This allows us to solve the stationary ACR by investigating the asymptotic stability of a general dt-LTI system, which can be examined by the well-known Jury stability criterion~\cite{jury1964theory}.
Consider a characteristic polynomial with variable $z\in \mathbb{R}$ in the following form,
\begin{equation*}
    D(z) = a_0 + a_1z + a_2z^2 + \ldots + a_Nz^N,
\end{equation*}
where $N \in \mathbb{N}^+$ is the degree of the characteristic polynomial and $a_i \in \mathbb{R}$, $i=1,2,\cdots,N$, are coefficients. The following tests determine whether the system represented by $D(z)$ has any pole outside the unit circle. A system must conform to all the following rules to be considered stable. 

\noindent Rule 1: If $z=1$, $D(z)>0$ must hold.

\noindent Rule 2: If $z=-1$, $z^N D(z) > 0$ must hold.    

\noindent Rule 3: $|a_0|<|a_N|$ must hold.

\noindent If all rules satisfied, we expand the Jury Array as follows.
\begin{center}
\begin{tabular}{ c c c c c c c }
\hline
1)      & $a_0$     & $a_1$    & $a_2$   & $a_3$     & \ldots     & $a_N$  \\ 
2)      & $a_N$     & \ldots   & $a_3$   & $a_2$     & $a_1$      & $a_0$  \\  
3)      & $b_0$     & $b_1$    & $b_2$   & \ldots    & $b_{N-1}$  & ~      \\
4)      & $b_{N-1}$ & \ldots   & $b_2$   & \ldots    & $b_0$      & ~      \\
\vdots  & \vdots    & \vdots   & \vdots  & ~         & ~          & ~      \\    
$2N-3$) & $v_0$     & $v_1$    & $v_2$   & ~         & ~          & ~      \\ 
\hline
\end{tabular}
\end{center}
We stop constructing further arrays once we reach a row with 2 members. Then, we use the following formula to calculate the values of the odd-number rows. 
{
\small
\begin{equation*}
b_k = 
\begin{vmatrix}
a_0 &\!\! a_{N-k} \\
a_N &\!\! a_k     \\
\end{vmatrix},
c_k = 
\begin{vmatrix}
b_0     &\!\! b_{N-1-k} \\
b_{N-1} &\!\! b_k     \\
\end{vmatrix},
d_k = 
\begin{vmatrix}
c_0     &\!\! c_{N-2-k} \\
c_{N-2} &\!\! c_k       \\
\end{vmatrix}.
\end{equation*}
}
The even number rows are equal to the previous row in reverse order. We will use $k$ as an arbitrary subscript value. These formulas are reusable for all elements in the array. This can be carried out to all lower rows of the array if needed.

\noindent Rule 4: Once the Jury array has been formed, all the following relationships must be satisfied until the last row of the array
\begin{equation*}
|b_0| > |b_{N-1}|,~~ |c_0| > |c_{N-2}|, ~~|d_0| > |d_{N-3}|.
\end{equation*}
The system is stable if all these conditions are satisfied.

\subsection{Proof of Theorem~\ref{th:theorem1}}

We can rewrite the recursive model~\eqref{eq:rec} as the following matrix-vector form
\begin{equation}
\bm{\xi}_{k} = \begin{bmatrix}
\bm{0}^{\top}      &    \bm{I}  \\
P_T & \bm{p}        \\
\end{bmatrix} \bm{\xi}_{k-1} + \bm{\beta},~\forall \, k \in \mathbb{N}^+,  
\label{eq:compact_system}
\end{equation}
where $\bm{I}$ is a $(T-1)$-dimensional identity matrix, $\bm{0} \in \mathbb{R}^{T-1}$ is a zero vector, and
\begin{equation*}
\begin{split}
\bm{\xi}_k \,&= \left[\, \E\!\left(\delta_{k+T-1}\right) ~ \ldots ~ \E\!\left(\delta_{k+1}\right) ~ \E\!\left(\delta_k\right)\, \right]^{\top}, \\
\bm{p} \,&= \left[\, P_{T-1} ~\ldots ~P_1 \, \right],~\bm{\beta} = [\,\bm{0}^{\top}  ~ 1\,]^\top,
\end{split}
\end{equation*}
with an initial condition
\begin{equation*}
\bm{\xi}_0 = \left[\, \E\!\left(\delta_{T-1}\right) ~ \ldots ~ \E\!\left(\delta_{1}\right) ~ \E\!\left(\delta_0\right)\, \right]^{\top},
\end{equation*}
Therefore, \eqref{eq:compact_system} can be recognized as a dt-LTI system, where $\bm{\xi}_k$, $k \in \mathbb{N}$, is the system state, $\bm{\beta}$ is the constant input, and $P_n$, $n=1,2,\cdots,T$, are the constant parameters. In this sense, the existence of the stationary ACR $\E\!\left( \delta_{\infty}\right)$ can be determined by the stability of the dt-LTI system using the Jury's criterion.

For any $z \in \mathbb{R}$, the characteristic polynomial~\eqref{eq:compact_system} is
\begin{equation}
D(z) = P_T + P_{T-1}z + \ldots + P_{1}z^{T-1} + z^T.  
\label{eq:poly_system}
\end{equation}
Given the polynomial (\ref{eq:poly_system}), we investigate the state convergence of the dt-LTI system \eqref{eq:compact_system} using the Jury's criterion recalled in Sec.~Appx-A. It is straightforward to verify that Rules 1-3 in Sec.~Appx-A hold for \eqref{eq:poly_system}. We then use the coefficients in \eqref{eq:poly_system} to construct the Jury array. The elements in the first row of the Jury array then become
\begin{equation*}
a_0 = P_T, \;\;a_1 = P_{T-1},\;\;\ldots,\;\;a_T = 1.    
\end{equation*}
Having the elements on row i and row 2 of Jury array, the elements $b_k$ and $b_{k+1}$ in the row 3 and row 4, with $k\in\{0,\ldots,T-1\}$, can be constructed as
\begin{equation*}
b_k = 
\begin{vmatrix}
a_0       & a_{T-k} \\
a_{T}   & a_k       \\
\end{vmatrix}
=a_0a_k - a_{T-k}a_{T},
\end{equation*}
\begin{equation*}
b_{k+1} = 
\begin{vmatrix}
a_0       &  a_{T-k-1}   \\
a_{T}   &  a_{k+1}     \\
\end{vmatrix}
=a_0a_{k+1} - a_{T-k-1}a_{T}.
\end{equation*}
From Lemma \ref{lem:transtion_inequality}, we obtain $0<a_0<a_1<\ldots<a_{T}=1$, which implies $a_0a_k < a_0a_{k+1}<a_{T-k-1}a_{T} < a_{T-k}a_{T}$. This inequality further implies $-1<b_k<b_{k+1}<0$ and $|b_0|>|b_{T-1}|$. These results reveal the relationship among the elements $b_k$. Similarly, we can construct $c_k$ and $c_{k+1}$, as
\begin{equation*}
c_k = 
\begin{vmatrix}
b_0      &  b_{T-1-k} \\
b_{T-1}  &  b_k       \\
\end{vmatrix}
=b_1b_k - b_{T+1-k} b_T,
\end{equation*}
\begin{equation*}
c_{k+1} = 
\begin{vmatrix}
b_0       &  b_{T-2-k}   \\
b_{T-1}   &  b_{k+1}     \\
\end{vmatrix}
=b_0b_{k+1} - b_{T-2-k} b_{T-1}.
\end{equation*}
Similarly, we can readily conclude that $1>c_k>c_{k+1}$, which also implies $|c_1|>|c_T-1|$. Similar analysis can be carried out to show that Rule 4 of Jury stability criteria always holds. Therefore, the characteristic polynomial~\eqref{eq:poly_system} meets Jury's stability criteria, which means that all the eigenvalues of the state transition matrix in~\eqref{eq:compact_system} are less than or equal to 1, i.e. the system presented in~\eqref{eq:compact_system} is asymptotically stable. This indicates that the limit $\E({\delta}_{\infty}) = \lim_{k \rightarrow \infty} \E(\delta_{k})$ exists. By taking the limit of both sides of~\eqref{eq:rec}, we obtain
\begin{equation*}
\textstyle \lim_{k \rightarrow \infty} \E\left(\delta_k \right) = 
1 - \sum^T_{n=1} P_{n} \lim_{k \rightarrow \infty} \E\left(\delta_{k-n} \right),
\end{equation*}
which leads to \eqref{eq:stationary_rate} and proves this Theorem~\ref{th:theorem1}. 

\subsection{Truncated Stochastic Variables}

As mentioned in Sec.~\ref{sec:tacr}, the side information $|e_{k-1}|< \eta$ in \eqref{eq:close} changes the support of $p_{\hat{e}_k}(\cdot)$ at every sampling instant $k$. With a constant threshold $\eta \in \mathbb{R}^+$, the change is specifically a symmetric truncation operation to the PDF $p_{\hat{e}_k}(\cdot)$.
Consider a scalar stochastic variable $\zeta \in \mathbb{R}$ with an infinite-support PDF $p_{\zeta}(\cdot)$. We use $\zeta^{[a,b]}$ to represent the truncated stochastic variable derived from $\zeta$ by trimming its support set with a fixed interval $\zeta \in [\,a,b\,]$, $a,b\in \mathbb{R}$. Due to this truncation operation, the derived variable $\zeta^{[a,b]}$ has different stochastic properties compared to its original $\zeta$. Specifically, its PDF reads
\begin{equation}
    p_{\zeta^{[a,b]}}\!\left(z\right) = \left\{ \begin{array}{ll}
    \rho_{\zeta}(a,b) p_{\zeta}(z), & a \leq z \leq b, \\
    0, & \mathrm{otherwise},
    \label{eq:truncation_PDF}
    \end{array} \right.
\end{equation}
where $z \in \mathbb{R}$ is an auxiliary variable, $p_{\zeta^{[a,b]}}(\cdot)$ denotes the PDF of $\zeta^{[a,b]}$ subject to a truncation interval $[\,a,b\,]$, and $\rho_{\zeta}(a,b)$ is a scalar calculated as $\rho_{\zeta}(a,b)= 1/(F_\zeta(b) - F_\zeta(a))$,
where $F_\zeta(\cdot)$ is the cumulative distribution function (CDF) of $\zeta$. 
Also, the expected value and the variance of $\zeta^{[a,b]}$ are
\begin{equation*}
\textstyle \E\!\left(\zeta^{[a,b]}\right) = \int^{b}_a z p_{\zeta^{[a,b]}}(z) \mathrm{d}z = \rho_{\zeta}(a,b) \int^{b}_a z p_{\zeta}(z) \mathrm{d}z,
\end{equation*}
\begin{equation*}
\begin{split}
\textstyle \mathrm{Var}\!\left(\zeta^{[a,b]}\right) =&\, \textstyle \int^{b}_a z^2 p_{\zeta^{[a,b]}}(z) \mathrm{d}z - \E^2\!\left(\zeta^{[a,b]}\right) \\
=&\, \textstyle  \rho_{\zeta}(a,b) \int^{b}_a z^2 p_{\zeta}(z) \mathrm{d}z - \E^2\!\left(\zeta^{[a,b]}\right).
\end{split}
\end{equation*}
Note that the difference between the mean values and the variances of a stochastic variable $\zeta$ and its truncated counterpart $\zeta^{[a,b]}$ is reflected not only by the additional multiplier $\rho_{\zeta}(a,b)$ but also by the changed upper and lower limits of the integrals, $a$ and $b$. If $\zeta$ is a Gaussian variable, $\zeta^{[a,b]}$ is not necessarily Gaussian. This means that all the properties that are proposed for Gaussian variables, such as the linear combination properties, may not hold for their truncated variables. Ignoring this effect may lead to the inaccurate characterization of the truncated stochastic variable. 

\subsection{Proof of Proposition~\ref{pp:trunc}}

If $\E(\zeta)=0$ and $p_{\zeta}(z) = p_{\zeta}(-z)$ hold, according to the definition of the PDF of truncated stochastic variables in~\eqref{eq:truncation_PDF}, we have
\begin{equation*}
p_{\zeta^{\eta}}(z) = \frac{p_{\zeta}(z)}{F_{\zeta}(\eta) - F_{\zeta}(-\eta)} = \frac{p_{\zeta}(-z)}{F_{\zeta}(\eta) - F_{\zeta}(-\eta)} = p_{{\zeta}^{\eta}}(-z).
\end{equation*}
Utilizing this property, we further have
\begin{equation*}
\E(\zeta^{\eta}) = \int_{-\eta}^{\eta} z p_{{\zeta}^{\eta}}(z) \mathrm{d} z = 0.
\end{equation*}
Therefore, condition 1) of Proposition~\ref{pp:trunc} is proved. Furthermore, the variance of $\zeta^{\eta}$ reads
\begin{equation*}
\begin{split}
\mathrm{Var}\!\left({\zeta}^{\eta} \right) = & \, \textstyle  \int_{-\eta}^{\eta} z^2p_{\zeta^{\eta}}(z)\mathrm{d} z - \E^2\!\left(\zeta^{\eta} \right)\\
= & \, \textstyle  \int_{-\eta}^{\eta}z^2p_{\zeta}(z)\mathrm{d}z \left/ {\int_{-\eta}^{\eta}p_{\zeta}(z)\mathrm{d}z} \right..
\end{split}
\end{equation*}
Note that $\mathrm{Var}\!\left( \zeta^{\eta} \right)$ is indeed a function of the truncation interval $\eta$. Thus, we represent it as $\mathrm{Var}(\eta)$. It can be verified that $\mathrm{Var}(\eta)$ is continuous and continuously differential for $\eta$. Moreover, we know 
\begin{equation*}
\lim_{\eta\rightarrow\infty}\mathrm{Var}(\eta) = \int_{-\infty}^{\infty}z^2p_{\zeta}(z)\mathrm{d} z = \mathrm{Var}(\zeta),~\lim_{\eta\rightarrow 0}\mathrm{Var}(\eta) = 0.
\end{equation*}
By taking the derivative of $\mathrm{Var}\!\left( \eta \right)$ to $\eta$, we obtain
\begin{equation*}
\begin{split}
\mathrm{Var}'(\eta)\, &  = \textstyle \left[\left( \int_{-\eta}^{\eta}z^2p_{\zeta}(z)\mathrm{d}z \right)'\int_{-\eta}^{\eta}p_{\zeta}(z)\mathrm{d}z \right. \\
&- \textstyle \left. \left( \int_{-\eta}^{\eta}p_{\zeta}(z)\mathrm{d}z \right)'\int_{-\eta}^{\eta}z^2p_{\zeta}(z)\mathrm{d}z \right] \!\left/\! {\left( \int_{-\eta}^{\eta}p_{\zeta}(z)\mathrm{d}z \right)^2} \right.\!.
\end{split}
\end{equation*}
Note that
\begin{equation*}
\textstyle \left( \int_{-\eta}^{\eta}z^2p_{\zeta}(z)\mathrm{d}z \right)' = \eta^2 \!\left[ p_{\zeta}(\eta)+p_{\zeta}(-\eta) \right] = 2 \eta^2 p_{\zeta}(\eta),
\end{equation*}
\begin{equation*}
\textstyle \left( \int_{-\eta}^{\eta}p_{\zeta}(z)\mathrm{d}z \right)' = p_{\zeta}(\eta)+p_{\zeta}(-\eta)  = 2 p_{\zeta}(\eta).
\end{equation*}
Thus,
\begin{equation*}
\mathrm{Var}'(\eta) = \left.{2p_{\zeta}(\eta)} \int_{-\eta}^{\eta} \left( \eta^2 - \zeta^2\right) p_{\zeta}(z)\mathrm{d} z \right/ {\int_{-\eta}^{\eta}p_{\zeta}(z)\mathrm{d} z}.
\end{equation*}
Since $p_{\zeta}(\cdot)$ is a non-negative, we conclude $V'(\eta) >0$, for all $\eta > 0$. This implies that $V(\eta)$ is a monotonically increasing function in the interval $\eta \in (0,\infty)$. Therefore, we can write $\mathrm{Var}(\zeta^{\eta}) = \mathrm{Var}(\eta) < \mathrm{Var}(\infty) = \mathrm{Var}(\zeta)$, for any $0<\eta<\infty$. Thus, condition 2) of Proposition~\ref{pp:trunc} is proved.

\subsection{Proof of Theorem~\ref{pp:e0}}

We prove Theorem~\ref{pp:e0} recursively. 
We first consider the case $k=1$. Since $\hat{e}_{1},e_1 \sim \mathcal{N}(0,\sigma)$, we have $p_{\hat{e}_1}(z) = p_{\hat{e}_1}(-z)$, for all $z \in \mathbb{R}$, $\E\!\left( \hat{e}_{1} \right) = \E(e_1) = 0$, and $\mathrm{Var}(\hat{e}_1) = \mathrm{Var}(e_1) = \sigma^2$. Then, for a truncated stochastic variable $\hat{e}_i^{\eta}$ with threshold $\eta > 0$, according to Proposition \ref{pp:trunc}, given any $i =1,2,\cdots,T-1$, such that $p_{\hat{e}_i}(z) = p_{\hat{e}_i}(-z)$, we have
$p_{\hat{e}_i^{\eta}}(z) = p_{\hat{e}_i^{\eta}}(-z)$, $\E(\hat{e}_i^{\eta}) =0$, and $\mathrm{Var}(\hat{e}_i^{\eta}) < \mathrm{Var}(\hat{e}_i)$.

According to~(\ref{eq:close}), we know $\hat{e}_{i+1} = A \hat{e}_{i}^{\eta} + w_i$, from which we conclude
\begin{equation*}
\textstyle p_{\hat{e}_{i+1}}\left( z \right) = \int_{-\eta}^{\eta} \;p_{\hat{e}_i^{\eta}} \left( \xi \right) p_w(z - A \xi) \mathrm{d}\xi.
\end{equation*}
Considering that $p_{\hat{e}_1^{\eta}}(\cdot)$ is an even PDF, and $p_w(z) = p_w(-z)$, for all $z \!\in \mathbb{R}$, we obtain
\begin{equation*}
\textstyle p_{\hat{e}_{i+1}}\left(z \right) = \int_{\xi=-\eta}^{\xi=\eta} \;p_{\hat{e}_i^{\eta}} \left(- \xi \right) p_w(-z + A \xi) \mathrm{d}\xi.
\end{equation*}
Set $\hat{z} = - \xi$, then we will have
\begin{equation*}
\begin{split}
p_{\hat{e}_{i+1}}\left(z \right)
&= \textstyle  \int_{-\hat{z}=-\eta}^{-\hat{z}=\eta} \;p_{\hat{e}_i^{\eta}} \left(\hat{z} \right) p_w(-z - A \hat{z}) \mathrm{d}(-\hat{z}) \\
&\textstyle  = -\int_{-\hat{z}=-\eta}^{-\hat{z}=\eta} \;p_{\hat{e}_i^{\eta}} \left(\hat{z} \right) p_w(-z - A \hat{z}) \mathrm{d}\hat{z} \\
&\textstyle  = \int_{\hat{z}=-\eta}^{\hat{z}=\eta} \;p_{\hat{e}_i^{\eta}} \left(\hat{z} \right) p_w(-z - A \hat{z}) \mathrm{d}\hat{z} 
=p_{\hat{e}_{i+1}}\!\left(-z \right).
\end{split}
\end{equation*}
This property leads to 
\begin{equation*}
\E(\hat{e}_{i+1}) = \int_{-\eta}^{\eta} z p_{\hat{e}_{i+1}}\!\left(z \right) \mathrm{d}z = 0.
\end{equation*}
Also, from Assumption \ref{as:as1}, we know that $e_i^{\eta}$ and $e_i$ are independent from $w_{i-1}$. Thus we conclude  
\begin{equation*}
\begin{split}
\mathrm{Var}(\hat{e}_{i+1}) &= A^2 \mathrm{Var}(\hat{e}^{\eta}_{i}) + \sigma^2,    \\
\mathrm{Var}(e_{i+1}) &= A^2\mathrm{Var}(e_i) + \sigma^2. \\
\end{split}
\end{equation*}
According to Proposition \ref{pp:trunc}, we have $\mathrm{Var}(\hat{e}_i^{\eta}) < \mathrm{Var}(\hat{e}_i)$. Therefore, for any $i$ such that $\mathrm{Var}(\hat{e}_i) \leq \mathrm{Var}(e_i)$, we have 
\begin{equation*}
\begin{split}
\mathrm{Var}(\hat{e}_{i+1}) &= A^2 \mathrm{Var}(\hat{e}_i^{\eta}) + \sigma^2 < A^2 \mathrm{Var}(\hat{e}_i) + \sigma^2 \\
&\leq A^2 \mathrm{Var}(e_i) + \sigma^2 = \mathrm{Var}(e_{i+1}).
\end{split}
\end{equation*}
Finally, we proved $p_{\hat{e}_i}(z) = p_{\hat{e}_i}(-z)$, $\E\!\left(\hat{e}_i \right) = 0$, and $\mathrm{Var}(\hat{e}_i) \leq \mathrm{Var}(e_i)$ hold for all $i=1,2,\cdots,T$. Note that $\mathrm{Var}(\hat{e}_i) = \mathrm{Var}(e_i)$ only when $i=1$.

\subsection{An Example of Computing ACR Analytically}

The computation procedure of $\E(\delta_k)$ for $k=1,2,3$ using the analytical method proposed in Sec.~\ref{sec:crc} is as follows. 

$\bullet$ For $k=1$, according to \eqref{eq:pdf_0_1}, we have 
\begin{equation*}
\textstyle \overline{P}_1 = \int_{-\eta}^{\eta}\delta(z)\mathrm{d}z = 1,~\mathrm{and}~\E(\delta_1) = 1- \overline{P}_1 = 0.
\end{equation*}

$\bullet$ For $k=2$, using~\eqref{eq:con2}, we can calculate
\begin{equation*}
\textstyle \overline{P}_2 = \int_{-\eta}^{\eta} p_{\hat{e}_1}\!\left( z \right) \mathrm{d} z = 0.6827,
\end{equation*}
where the analytical form of $p_{\hat{e}_1}(\cdot)$ is provided in~\eqref{eq:pdf_0_1}. Then, according to \eqref{eq:expcon}, we have
\begin{equation*}
\E(\delta_2) = 1 - \overline{P}_1 \E(\delta_1) - \overline{P}_2 \overline{P}_1 \E(\delta_0) = 0.3173.
\end{equation*}

$\bullet$ For $k=3$, we have $\hat{e}_2 = A \hat{e}_1^{\eta} + w_1$ according to \eqref{eq:close}. Thus, the analytical form of $p_{\hat{e}_2}(\cdot)$ reads
\begin{equation}\label{eq:p2z2}
\textstyle p_{\hat{e}_2}(z) = \int^{\eta}_{-\eta} p_{\hat{e}_1^{\eta}}(\xi)\;p_w(z-A\xi) \mathrm{d} \xi.
\end{equation}
Note that $p_{\hat{e}_1^{\eta}}(\cdot)$ is a truncated Gaussian PDF and $p_{w}(\cdot)$ is a Gaussian PDF, which makes the calculation of $p_{\hat{e}_2}(\cdot)$ nontrivial. According to the formulations in Appx.-C, we have
\begin{equation}\label{eq:erf}
p_{\hat{e}_1^{\eta}}\!\left(z\right) =  \frac{1}{\sqrt{2\pi}\sigma \;\mathrm{erf}\left(\frac{\eta}{\sqrt{2\sigma^2}} \right)} \;\mathrm{exp}\!\left(-\frac{z^2}{2 \sigma^2}\right),
\end{equation}
where $\mathrm{erf}(x)=\frac{2}{\sqrt{\pi}} \int^x_0 \mathrm{exp}(-x^2) \mathrm{d}x$ is the Gaussian error function. Substituting \eqref{eq:erf} to the integral term in (\ref{eq:p2z2}), we get
\begin{align}\nonumber
p_{\hat{e}_2}(z) =&\, \int^{\eta}_{-\eta} \frac{\mathrm{exp} \left( -\frac{\left(z-A \xi\right)^2+ {\xi^2}} {2 \sigma^2 }\right)}{2\pi \sigma^2 \mathrm{erf}\!\left(\frac{\eta}{\sqrt{2}\sigma} \right)}  \mathrm{d}\xi \\\nonumber
=& \frac{\mathrm{exp} \left( -\frac{z^2}{2 \sigma^2 \left(A^2+1\right)} \right) }{2\pi \sigma^2 \mathrm{erf}\!\left(\frac{\eta}{\sqrt{2}\sigma} \right)}\! \int^{\eta}_{-\eta}\! \mathrm{exp}\! \left( -\frac{\left(\xi-\frac{A z}{A^2+1}\right)^2}{2 \bar{\sigma}^2}  \right)\! \mathrm{d}\xi 
\\=& \frac{\mathrm{exp}\left(  -\frac{z^2}{2 \sigma^2 \left(A^2 + 1\right)}\right) }{2\sigma \sqrt{2\pi\!\left(A^2 + 1 \right)} \mathrm{erf}\!\left(\frac{\eta}{\sqrt{2}\sigma} \right)} \label{eq:truncated_error_distribution}\\\nonumber
\times& \left\{\mathrm{erf}\!\left( \frac{\eta A^2+\eta  - A z}{\sqrt{2} \sigma \sqrt{A^2 + 1}} \right)\right.
\!+\!\left. \mathrm{erf}\!\left( \frac{\eta A^2+\eta  + A z}{\sqrt{2} \sigma \sqrt{A^2 + 1}} \right) \right\}.
\end{align}
Therefore, we calculate
\begin{equation*}
\overline{P}_3 = 1 -\int_{-\eta}^{\eta} p_{\hat{e}_2}(z) \mathrm{d}z = 0.5872.
\end{equation*}
According to \eqref{eq:expcon}, we have
\begin{equation*}
\E (\delta_3) = \, 1 -\overline{P}_1 \E (\delta_2) - \overline{P}_2 \overline{P}_1 \E (\delta_1) - \overline{P}_3 \overline{P}_2 \overline{P}_1 \E (\delta_0) =0.2818.
\end{equation*}

It has been noticed that the analytical form of $p_{\hat{e}_2}(\cdot)$ in \eqref{eq:truncated_error_distribution} becomes very complicated. The computation of $\E(\delta_k)$ for $k>3$ is even more difficult due to the complicated form of $p_{\hat{e}_{k-1}}(\cdot)$. Therefore, we only provide the results for $k \leq 3$.

\bibliographystyle{ieeetr}
\bibliography{reference}

\balance
\end{document}